\documentclass[12pt]{article}
\usepackage{xr}
\usepackage{amssymb}
\usepackage{amsfonts}
\usepackage{amsmath}
\usepackage{theorem}
\usepackage{graphicx}
\usepackage{xcolor}
\usepackage[authoryear,round,longnamesfirst]{natbib}
\usepackage{geometry}

\newtheorem{condition}{Condition}

\newtheorem{proposition}{Proposition}
\newtheorem{remark}{Remark}

\newenvironment{proof}[1][Proof]{\textbf{#1.} }{\ \rule{0.5em}{0.5em}}
\renewcommand{\cite}{\citeasnoun}
\allowdisplaybreaks

\begin{document}

\title{ Non-Existent Moments of Earnings Growth\thanks{\footnotesize\setlength{\baselineskip}{4.2mm} First arXiv date: March 15, 2022. We gratefully acknowledge the UK Data Service for granting access to NESPD and their guidance and assistance with regard to this data set. 
We benefited from useful comments by St\'ephane Bonhomme, Giuseppe Cavaliere, Eva F. Janssens, Gonzalo Paz-Pardo, Lorenzo Trapani, and seminar participants at University of Bologna and University of Pittsburgh.
This research was partially supported by the Italian Ministry of University and Research (PRIN Grant 2017TA7TYC).
All remaining errors are ours.
 \smallskip}}
\author{
Silvia Sarpietro\thanks{\footnotesize\setlength{\baselineskip}{4.2mm} Silvia Sarpietro: \texttt{silvia.sarpietro@unibo.it}. Department of Economics, University of Bologna, 40126 Bologna BO, Italy\smallskip}\\Bologna
\and
Yuya Sasaki\thanks{\footnotesize\setlength{\baselineskip}{4.2mm} Yuya Sasaki: \texttt{yuya.sasaki@vanderbilt.edu}. Department of Economics, Vanderbilt University, VU Station B \#351819, 2301 Vanderbilt Place, Nashville, TN 37235-1819, USA\smallskip}\\Vanderbilt
\and
Yulong Wang\thanks{\footnotesize\setlength{\baselineskip}{4.2mm} Yulong Wang: \texttt{ywang402@syr.edu}. Department of Economics, Syracuse University, 110 Eggers Hall, Syracuse, NY 13244-1020, USA}\\Syracuse
}
\date{}
\maketitle

\begin{abstract}\setlength{\baselineskip}{4.9mm}

The literature often employs moment-based earnings risk measures like variance, skewness, and kurtosis. However, under heavy-tailed distributions, these moments may not exist in the population. Our empirical analysis reveals that population kurtosis, skewness, and variance often do not exist for the conditional distribution of earnings growth. This challenges moment-based analyses. We propose robust conditional Pareto exponents as novel earnings risk measures, developing estimation and inference methods. Using the UK New Earnings Survey Panel Dataset (NESPD) and US Panel Study of Income Dynamics (PSID), we find: 1) Moments often fail to exist; 2) Earnings risk increases over the life cycle; 3) Job stayers face higher earnings risk; 4) These patterns persist during the 2007--2008 recession and the 2015--2016 positive growth period.




\vfill
\end{abstract}

\section{Introduction\label{sec:intro}}

One of the main goals of the literature on income dynamics is to quantify income risk. 
Recent studies show that the distribution of income changes has heavier tails than the normal distribution. 
This evidence has relevant implications for how income risk affects economic decisions.\footnote{See, for instance, \citet*{golosov2016redistribution} and references therein.} 
Since heavy tails incur enormous risk premium costs in economies with risk-averse agents, quantifying the tail heaviness, along with features of preferences of economic agents, is necessary to measure the social costs of the uncertainties. 

However, the existing literature relies primarily on moment-based measures of income risk.
For instance, the recent influential paper by \citet*[][]{guvenen2021data} uses skewness and kurtosis to characterize the distributions of individual earnings dynamics in the U.S. and provide their policy implications. They target this set of moments in their proposed procedure to estimate the parameters of a flexible stochastic process of earnings. The estimated process further feeds back into economic research and policy analyses, for instance in the analysis of life-cycle incomplete market models.
While \textit{sample} skewness and \textit{sample} kurtosis are generally finite, their \textit{true population} counterparts may not exist when the distribution has a heavy tail.
If the population moments do not exist, then their sample counterparts can be arbitrarily large and hence misleading. 

In this paper, we address this issue by first proposing a method of estimation and inference about the conditional Pareto exponent, which characterizes the tail heaviness. 
Second and more importantly, applying our method to the administrative data set for the UK, the New Earnings Survey Panel Dataset (NESPD), and the US Panel Study of Income Dynamics (PSID), we quantify the tail heaviness in the distribution of earnings changes and discuss how it varies with heterogeneous attributes of individuals. 
We find that the commonly used measures of income risk, such as sample kurtosis, skewness, and even standard deviation, may be misleading since their population counterparts may not exist (i.e., may be infinite) for the conditional distribution of earnings growth given age, gender, and past earnings. 
Finally, by interpreting the conditional Pareto exponent as a robust measure of the conditional earnings risk, we show that tail earnings risk increases over the life cycle, is higher for job-stayers, and these patterns appear in both the period 2007--2008 of the great recession and the period 2015--2016 of a positive growth despite some differences.

There is extensive literature on earnings and income dynamics models, dating back to the 1970s
-- see the comprehensive survey by \citet*{moffitt2018income}.
To explain the heavy tail feature, the existing literature exploits the following strategies: 
1) fitting data to mixed normal distributions using mixture approaches \citep*[][]{geweke2000empirical,bonhomme2009assessing}; 
2) conducting nonparametric density estimates and graphically exhibiting heavier tails than normal distributions with deconvolution and spectral decomposition approaches \citep*[][]{horowitz1996semiparametric,bonhomme2010generalized,arellano2017earnings,botosaru2018nonparametric,hu2019semiparametric};
3) reporting the sample moments such as skewness and kurtosis and comparing them with counterparts from normal distribution \citep*[][among others]{guvenen2021data,arellano2021}; 
4) reporting some quantile-based measures such as Kelly's measure of skewness and Crow-Siddiqui measure of kurtosis \citep*[][among others]{guvenen2021data}; and
5) reporting the mean absolute deviation \citep*[][]{arellano2021}.

The first approach, based on mixed normal distribution, essentially assumes an exponentially decaying tail, which cannot characterize tail heaviness.
The second one is based on nonparametric density estimation, which tends to behave poorly in the tails due to few extreme observations.
The third approach based on sample moments may yield misleading information if the population moment does not exist due to heavy tails.
In particular, the \textit{population} kurtosis may not exist even if a researcher can always get a finite \textit{sample} kurtosis. 
We indeed find that the population kurtosis, skewness, and even variance may not exist.
The fourth one is robust to heavy tails: these non-extremal quantile-based measures are always well-defined even if the population moments do not exist. 
However, they fail to capture extreme income changes and tail events. 
Among other things, information on tail events is central to quantifying earnings risk under heavy-tailed distributions.
Like the fourth one, the fifth approach based on the first moment alone is not informative about extreme income changes and tail events.

Given the above concerns about existing approaches, we aim for an alternative measure of tail heaviness to characterize the earnings risk conditional on individuals' attributes. 
The conditional Pareto exponent is proposed as our robust measure of tail heaviness.
We have the elegant characterization that the conditional kurtosis/skewness/variance is infinite if and only if the conditional Pareto exponent is smaller than the value of four/three/two.
We propose methods of estimation and inference about the Pareto exponent in the conditional distribution of earnings risk given observed attributes, such as gender, age, and base-year income level. 
Unlike the aforementioned approaches, this new method can be used to characterize the tail heaviness robustly even if the rate of tail decay is slower than exponential ones and even if the population kurtosis, skewness, or variance is infinite. 
We emphasize that we do \textit{not} make the assumption that data follow a Pareto distribution. Instead, we only assume the regularly varying (RV) tail condition, which nonparametrically encompasses a broad class of distributions, including Pareto, Student-t, Cauchy, and F distributions among many others. Numerous recent papers document that the distribution of income/earnings, both in levels and growth, roughly follows the power law, which necessarily entails ‘extreme’ observations and endorses our assumption -- see for instance \citet*{guvenen2021data} for earnings growth and \citet*{guvenen2022global} for additional empirical evidence that supports our assumption.\footnote{Based on the evidence provided by all papers submitted within the Global Income Dynamics Projects (GRID), \citet*{guvenen2022global} report that the distribution of income growth has thick Pareto tails in all 13 countries currently under the GRID.}
The RV tail condition is in line with these observations from the literature.

To illustrate the empirical significance, we apply the proposed method to two data sets: the UK NESPD and the US PSID, with our focus on the former. 
NESPD is an employer-based survey data set on individual earnings in the UK and has been investigated very recently \citep*[cf.][]{de2021wage,bell2021time}. 
We contribute to this new literature with the following three findings.
First, we find that the distribution of the earnings changes (measured by the difference between the log earnings across two subsequent years) conditional on age, gender, and past earnings exhibits substantially heavy tails. 
The conditional Pareto exponent is significantly less than four for the majority of the subpopulations, raising the concern that the kurtosis may be infinite. 
Also, the conditional Pareto exponent is significantly less than three or even two for some subpopulations, and we thus reject the hypothesis of finite conditional skewness and even finite conditional variance.
These results are robust across years, including the period of recession in 2007-2008 and the period of positive growth in 2015-2016 among others.

Second, we find remarkable patterns that the tail earnings risk is increasing over the life cycle. 
This is documented as that 40- and 50-year-old workers face higher earnings risk than 30-year-old workers. 
Third, we quantify the tail heaviness of the distribution of earnings changes for job-stayers and find that, at low earnings levels, conditional kurtosis does not exist, contrary to what we observe for the whole population. Excluding middle-aged men at top quantiles of earnings, we reject the hypothesis of finite conditional kurtosis for job-stayers at all earnings levels.
In Appendix \ref{sec:psid}, we also apply our method to the US PSID and obtain similar empirical findings on the tail heaviness of the conditional distribution of income risk. 

The above findings closely compare to the findings in \citet*{guvenen2021data} and the income literature inspired by their methodology.
In particular, \citet*{guvenen2021data} find in the US data set that the conditional variance of earnings risk is higher for younger workers at the bottom of the income distribution, and the kurtosis increases with age and with lagged income up to the top 5\% of the income distribution where it sharply declines. 
Since the data exhibit substantially heavier tails than those of the normal distribution, these sample moments could be less informative if their population counterparts are indeed infinite.
Our results confirm theirs, but now with our new measure of earnings risk being robust to infinite moments.
On the other hand, our findings differ from those in \citet*{arellano2021}, who examine the Spanish administrative data, and find that the income risk measured by the conditional mean absolute deviations is inversely related to income and age, and the income risk inequality increases markedly in the recession. 
The differences can be explained by a number of factors, such as different countries and different sample selection criteria as well as different measures of income risk for different research objects.


{\bf Organization:}
The rest of the paper is organized as follows.
Section \ref{sec:preview} describes the data to be examined in this article and provides some descriptive analysis of the heavy-tail phenomena in the earnings growth risk.
Section \ref{sec:econometric_method} formally introduces our econometric method, and Section \ref{sec:empirical} applies it to the UK data set.
Section \ref{sec:conclusion} concludes with some remarks.
The Appendix collects asymptotic theories, computational details, Monte Carlo simulations, additional empirical results with the UK and US data sets, and mathematical proofs. 

\section{Data and Preview\label{sec:preview}}

\subsection{New Earnings Survey Panel Data}\label{sec:admin}

We start with introducing the New Earnings Survey Panel Data (NESPD), an administrative data set at the individual level on UK earnings from the UK Social Security. It is an annual panel running from 1975 to 2016, and it surveys around 1\% of the UK workforce. All employees whose National Insurance Number (NIN) ends in a given pair of digits are included in the survey. The NIN number is randomly issued to all UK residents at age 16 and kept constant throughout the lifetime of an individual. The NESPD is a survey directed to all employers whose employees qualify for the sample: the employers complete the questionnaire based on payroll records for their employees. As a result of being directed to the employer, NESPD has a low non-classical measurement error. 
 
The survey reports the employees' gender, age, and detailed work-related information: annual, weekly, and hourly earnings, hours of work, occupation, industry, working area, firms' number of employers, and unionization. This information relates to a specified week in April of each year: the data sample is taken on the first day of April of each calendar year and concerns complete employee records only. The NESPD contains complete information on the employees' working life from the first year they started working until retirement age over the years 1975-2015, as long as the employer answered the questionnaire and the individual was working with the last recorded employer in April.

Given that NESPD contains detailed information on earnings and a long panel component, it has been used for a broad range of topics in the literature: among the others, \citet{goos2007lousy} have used NESPD to document job polarization in the UK;  \citet{nickell2003nominal} and \citet{elsby2016wage} have employed this study for analysis of wage rigidities; \citet{adam201935} provide an analysis of valid response rates in NESDP. However, the study of income and earnings dynamics in the UK has not received lots of attention, with a couple of notable exceptions in recent years: \citet{de2021wage} and \citet{bell2021time} are two main references for earnings dynamics in the UK using NESPD for the analysis.

The information in NESPD is of exceptionally high quality. However, it is worth mentioning that not all workers are covered and there are non-negligible non-response issues. 
In particular, data in NESPD are collected in a specified week of April. 
As a result, NESPD might under-sample part-time workers if their weekly earnings fall below the threshold for paying National Insurance and those who moved jobs recently. 
Furthermore, the data set is unbalanced with possibly non-random missing observations. This issue might be problematic for the representativeness of those at the bottom of the distribution and with unstable spells. 
We refer to \citet{bell2021time} and \citet{de2021wage} for a deeper investigation of the data limitations of NESPD.

Following the earnings literature, we further extract a subset of NESPD, which corresponds to employees with a strong attachment to the labor force. 
Specifically, we follow \citet{de2021wage} to use the following selection criteria: we drop the observations below 5\% of the median earnings (roughly \pounds1,300 a year), 
individuals whose total working hours exceed 80 hours per week, and individuals that display negative values in earnings or hourly wages; we do not consider individuals whose hours worked or weekly pay are missing. Earnings are deflated using CPI (2015=100). The earnings measure is the residual obtained by regressing the logarithm of earnings on year and age dummies. 
We extract a portion of data for a pair of years across the period of the Great Recession, 2007-2008. The above-described sample selection leaves 78,531 individuals.

For the purpose of comparisons, we also use the portion of data for the pair, 2015 and 2016, of the most recent survey years.
Note that this period is associated with positive economic growth in the UK.
The sample selection procedure described above leaves 95,906 individuals for this period. 
The results for other years are similar and hence postponed to Appendix \ref{sec:additional_results} for readability.

Apart from age, gender, and income, we could further consider industry and occupation as heterogeneous attributes in the analysis. NESPD contains information on both industry and occupation with the Standard Industry Classification (SIC) and Standard Occupation Classification (SOC) codes.

We define a measure $Y$ as the difference between the log earnings in 2007 and the log earnings in 2008 (i.e., one-year earnings growth rate).
Similarly, we also construct this variable for the period between 2015 and 2016.
Figures \ref{fig:nespd_density_2007} and \ref{fig:nespd_density_2015}  show the kernel density estimates of the measure $Y$ for the period 2007--2008 and the period 2015--2016, respectively.\footnote{The kernel estimates in these figures are obtained with the Gaussian smoothing kernel. Given the concerns raised by Degiannakis et al. (2023b) and Lahr (2014) that, under the existence of extreme values, the visual output of the kernel estimates may vary significantly under different weighting functions, we check various kernel functions (e.g., Epanechnikov, rectangular, and triangular kernels) in Appendix \ref{app: density}. Our analysis reveals no substantial differences among them.}
In each figure, the left and right panels show the densities for men and women, respectively. The figures illustrate clear departures from normality: each kernel density exhibits a large spike in the middle of the distribution sticking upward out of the reference normal density.
Moreover, each kernel density has heavier tails compared to the reference normal density.
These features of the estimated densities suggest that the actual distributions of $Y$ indeed have heavier tails than normal distributions, as documented in the previous literature. 

We also inspect the sample moments of the distribution of one-year earnings changes conditional on previous earnings. 
Figure \ref{fig:nespd_moments} shows the standard deviation, the Kelly's measure of skewness, and the Crow-Siddiqui measure of kurtosis, respectively, of earnings changes by previous earnings, for men (left three figures) and women (right three figures). 
The Kelly's measure of skewness is defined as $(Q_{.9}Y-Q_{.5}Y)-(Q_{.5}Y-Q_{.1}Y)/(Q_{.9}Y-Q_{.1}Y)$ and the Crow-Siddiqui measure of kurtosis as $(Q_{.975}Y-Q_{.025}Y)/(Q_{.75}Y-Q_{.25}Y)$, where $Q_{\tau}Y$ denotes the $\tau$-quantile of the distribution of $Y$.
As in \citet{de2021wage}, we find that the Crow-Siddiqui measure of kurtosis is higher than the reference value of 2.91 for the normal distribution, and the Kelly's measure of skewness deviates from zero, ranging from positive values for low earnings groups to negative values for higher values of previous earnings. 
However, as mentioned in the introduction, these quantile-based measures do not take into account the very tail observations that are above, say $Q_{.975}Y$. 
These observations are indeed informative about the tail feature of the earnings risk and inequality. 
In comparison, our proposed method takes all samples into consideration and is explicitly designed to characterize the tail heaviness. 

In the rest of the paper, we analyze the heaviness of the tails of the conditional distributions of $Y$ given the earnings level and age in the base year (i.e., the base year is 2007 for the period 2007--2008 and it is 2015 for the period 2015--2016).
Specifically, the heaviness of the tails is quantified by the conditional Pareto exponent, and we propose an econometric method of estimation and inference about this measure of tail heaviness.
Its value, in particular, informs whether the $r$-th conditional moments of the income growth exist for $r=2,3,4$, and so on.
If the conditional Pareto exponent, $\alpha(x_0)$, given $X=x_0$ is less than $r$, then the conditional distribution of $Y$ given $X=x_0$ does not have a finite $r$-th moment. 
Section \ref{sec:econometric_method} introduces a method of estimation and inference, and Section \ref{sec:empirical} presents the empirical results.
It turns out that we reject the hypotheses of finite conditional kurtosis, finite conditional skewness, and even finite conditional standard deviation for some subpopulations.

\begin{figure}[h]
	\centering
		\includegraphics[width=0.49\textwidth]{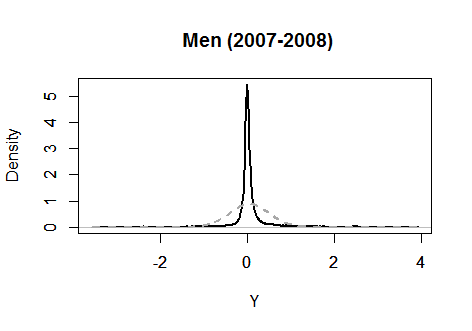}
		\includegraphics[width=0.49\textwidth]{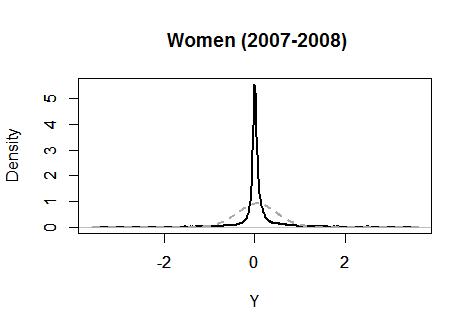}
	\caption{Kernel density estimates (black line) of $Y$, which is defined as a one-year change in log earnings, in 2007 in the NESPD. The left (respectively, right) panel shows the density of men (respectively, women). Also shown in gray dashed lines are the normal density fit to data. Number of individuals: 38,955 men (39,576 women).}
	\label{fig:nespd_density_2007}
\end{figure}

\begin{figure}[h]
	\centering
		\includegraphics[width=0.49\textwidth]{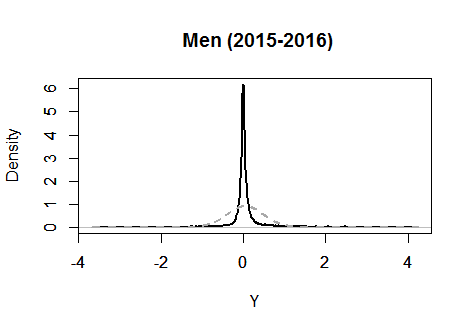}
		\includegraphics[width=0.49\textwidth]{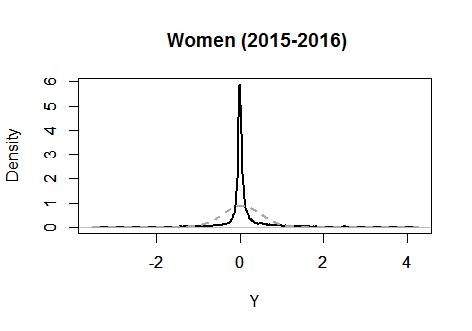}
	\caption{Kernel density estimates (black line) of $Y$, which is defined as a one-year change in log earnings, in 2015 in the NESPD. The left (respectively, right) panel shows the density of men (respectively, women). Also shown in gray dashed lines are the normal density fit to data. Number of individuals: 45,985 men (49,921 women).}
	\label{fig:nespd_density_2015}
\end{figure}

\begin{figure}
	\centering
		\includegraphics[width=0.45\textwidth]{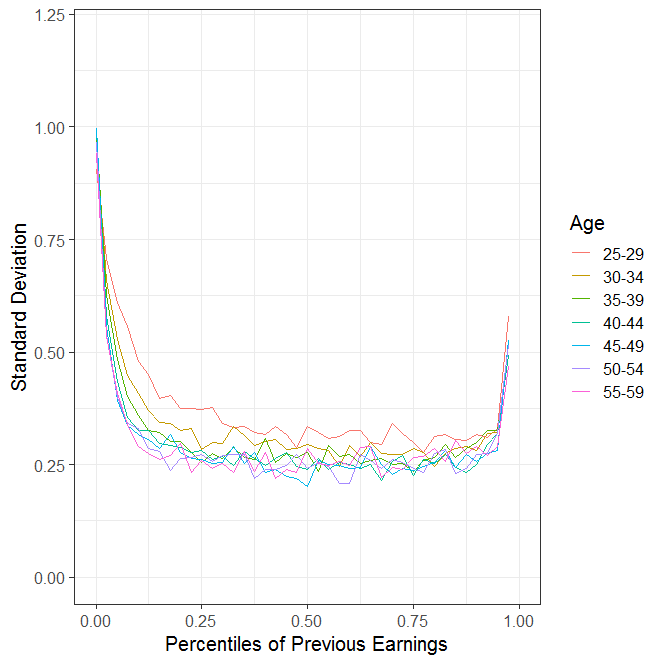}
		\includegraphics[width=0.45\textwidth]{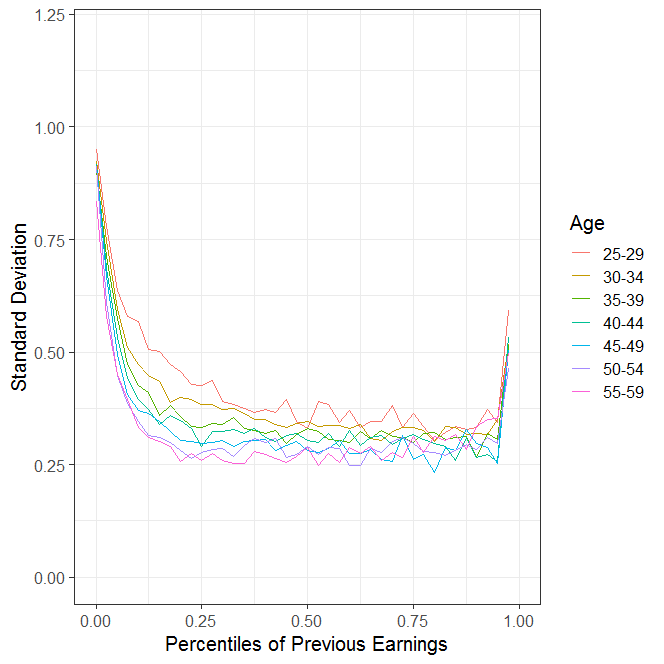}
		\includegraphics[width=0.45\textwidth]{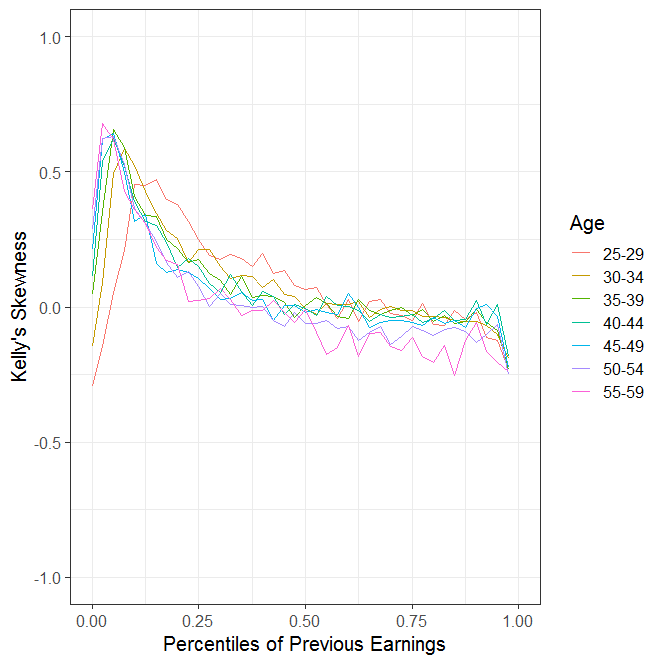}
		\includegraphics[width=0.45\textwidth]{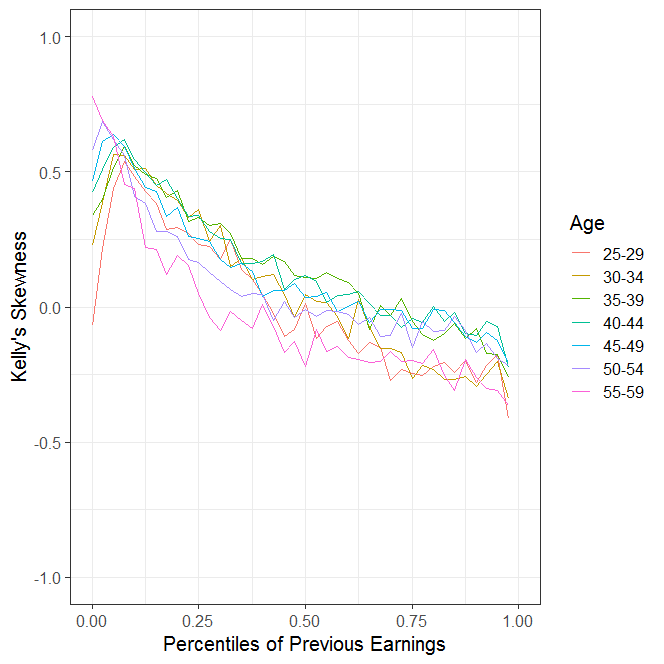}
		\includegraphics[width=0.45\textwidth]{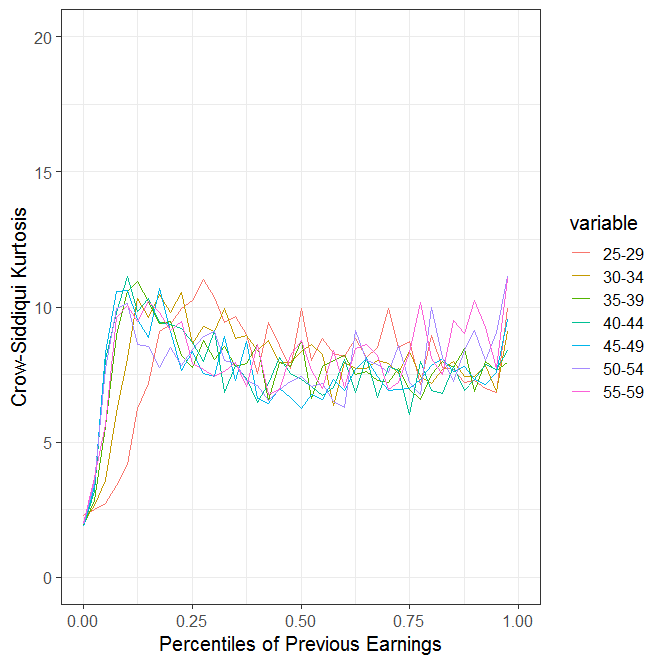}
		\includegraphics[width=0.45\textwidth]{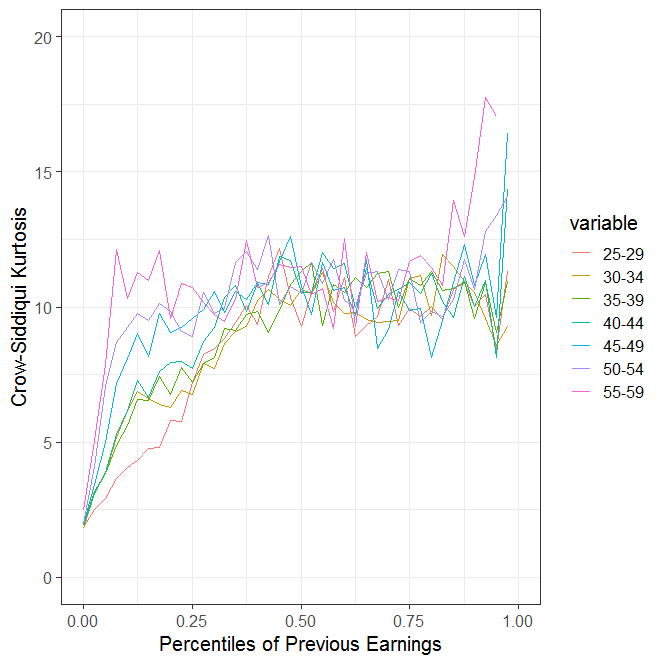}
	\caption{Moments of one-year log earnings changes by previous earnings, for men (left) and women (right) based on the NESPD.}
	\label{fig:nespd_moments}
\end{figure}

\section{Measurement of Conditional Tail Risk\label{sec:econometric}}\label{sec:econometric_method}

We now introduce our proposed measure of conditional tail risk, and present the method of estimation and inference.
Formal theoretical justifications are relegated to the Appendix.
Let $Y$ denote the variable of interest and let $X$ denote a vector of individual characteristics, such as the income level and age in the base year. 
We first assume that the sample is i.i.d. so that the joint distribution $F_{Y,X}\left(\cdot ,\cdot \right)$ of $(Y,X)$ is unique. 

\medskip\noindent
\textbf{Assumption 1:} \textit{The sample $(Y,X)$ is i.i.d.} 
\medskip

In our application to earnings changes, the cross-sectional i.i.d. assumption is typically imposed explicitly or implicitly. While the literature has emphasized the relevance of time dependence with ARCH effects or stochastic volatility, for several questions ranging from consumption insurance to income mobility, \citep*[see][]{meghir2004income}, it is plausible that the i.i.d. assumption is satisfied when it comes to cross-sectional data of earnings changes. 
Also, note that this is a sufficient condition that can be relaxed to allow for some form of weak dependence. 
See Appendix \ref{sec:additional_details} for related econometric methods that allow for time series data.


Consider a rectangular array $\{(Y_{ij},X_{ij}): i \in \{1,\cdots,I\}, j \in \{1, \cdots, J\}\}$. 
Such a data structure can be constructed by randomly splitting a cross-sectional data set of size $N$ into an $I \times J$ array such that $I \cdot J \approx N$, which is what we implement in Section \ref{sec:empirical}.
See Appendix \ref{sec:choice} for how to choose $I$ and $J$ in practice under such a construction of an array.

Our second main assumption is that the conditional distribution $F_{Y|X=x_{0}}\left( \cdot \right) $ is regularly varying at infinity, which is essentially equivalent to
\begin{equation*}
1-F_{Y|X=x_{0}}(y)\propto y^{-\alpha (x_{0})}\mathcal{\ }\text{as }%
y\rightarrow \infty \text{,} 
\end{equation*}
where $\alpha \left( x_{0}\right) >0$ denotes the $x_0$-conditional Pareto exponent that characterizes the tail heaviness of the conditional distribution of $Y$ given $X=x_0$. 
A formal definition of regular variation is as follows. 
	
\medskip\noindent
	\textbf{Assumption 2:} \textit{ The conditional distribution $F_{Y|X=x_{0}}\left( \cdot \right) $ 	is regularly varying at infinity, that is, for all $y>0$
	\begin{equation*}
		\frac{1-F_{Y|X=x_{0}}(yt)} {1-F_{Y|X=x_{0}}(y)}  \rightarrow y^{-\alpha (x_{0})}\mathcal{\ }\text{as }%
		t\rightarrow \infty \text{,} 
	\end{equation*}
	where $\alpha \left( x_{0}\right) >0$ denotes the $x_0$-conditional Pareto exponent that characterizes the tail heaviness of the conditional distribution of $Y$ given $X=x_0$.  }
\medskip
	
We again emphasize that we do \textit{not} impose the assumption that data follow a Pareto distribution.
Instead, this assumption only requires the regularly varying (RV) tail condition. This RV tail condition is mild and satisfied by many commonly used distributions, such as Student-t, F, and gamma distributions among numerous others, as well as nonparametric distributions. 
Moreover, it has been widely documented that the earnings data sets exhibit a Pareto-type (i.e., RV) tail and hence support this condition \citep*[cf.][]{Gabaix2009, gabaix2016power, guvenen2021data}. See a survey by \citet*{gabaix2016power} and \citet*{guvenen2022global} for additional empirical evidence that supports our assumption. From this existing literature, the assumption of the existence of extreme events is legitimate for the distribution of income and earnings growth. In line with these findings, our data provide suggestive evidence that extremes can occur and, thus, further validate our assumption, as shown in the densities in Figures \ref{fig:nespd_density_2007}--\ref{fig:nespd_density_2015}, and in the Pareto plots in Figures \ref{fig:Paretoplot_cond}--\ref{fig:pareto_plot_K30}.
See Appendix \ref{sec:primitive_conditions} for more discussions and primitive conditions.

With the parameter $\alpha\left(x_{0}\right)$, we can characterize the existence of the conditional moments as follows.
For any $r\in \mathbb{R}^{+}$, 
\begin{align*}
\mathbb{E}\left[ Y_{ij}^{r}|X_{ij}=x_{0}\right] <&\infty \text{ if }\alpha
(x_{0})>r 
\qquad\text{and}\\
\mathbb{E}\left[ Y_{ij}^{r}|X_{ij}=x_{0}\right] =&\infty \text{ if }\alpha
(x_{0})<r.
\end{align*}
Thus, the test of finite conditional $r$-th moment can be represented by the competing hypotheses
\begin{equation}
H_{0}:\alpha (x_{0})>r\text{ against }H_{1}:\alpha (x_{0})\leq r.
\label{hypo}
\end{equation}
In particular, we set $r=2$, $3$, and $4$ for tests of the existence of the standard deviation, skewness, and kurtosis, respectively.

To construct a feasible test for \eqref{hypo}, we need to obtain a random sample from the conditional distribution $F_{Y|X=x_{0}}\left( \cdot \right)$. 
This would be straightforward if $X$ is discrete. 
However, random sampling from such a conditional distribution is infeasible when $X$ includes at least one continuous random variable, such
as the income level, which we use in our data analysis. 
To overcome this issue, we propose the following procedure to extract the local sample. 
Let $\left\vert \left\vert \cdot \right\vert \right\vert $ denote the Euclidean norm.

\begin{enumerate}
\item For each $i$, find the nearest neighbor (NN) of $\{X_{ij}\}_{j=1}^{J}$
to $x_{0}$ and the induced $Y_{ij}$ associated with the NN, that is $%
Y_{ij}=Y_{ij^{\ast }}$ where $\left\vert \left\vert X_{ij^{\ast
}}-x_{0}\right\vert \right\vert =\min_{j\in \{1,...,J\}}\left\vert
\left\vert X_{ij}-x_{0}\right\vert \right\vert$. Denote the induced $Y$ by $%
\{Y_{i,[x_{0}]}\}_{i=1}^{I}$.

\item Sort $\{Y_{i,[x_{0}]}\}_{i=1}^{I}$ descendingly as $%
\{Y_{(1),[x_{0}]}\geq Y_{(2),[x_{0}]}\geq ,...,\geq Y_{(I),[x_{0}]}\}$.
Collect the largest $K+1$ of them as 
\begin{equation*}
\mathbf{Y}(x_{0})=\{Y_{(1),[x_{0}]},Y_{(2),[x_{0}]},...,Y_{(K+1),[x_{0}]}\}.
\end{equation*}

\item Use $\mathbf{Y(}x_{0}\mathbf{)}$ to estimate the Pareto exponent $\alpha\left(x_{0}\right)$ by
the formla in \eqref{eq:hill} below.
\end{enumerate}

We postpone until Appendix \ref{sec:primitive_conditions} formal discussions of conditions under which this procedure works and why it works in theory.
Here, we instead discuss its intuition.
First, for each $i$, we select the NN among $\{X_{ij}\}_{j=1}^{J}$ to $x_{0}$. 
When $J$ (as well as $I$) is large enough, such a NN gets close enough to $x_{0}$ and therefore its induced value $Y_{i,[x_{0}]}$ behaves as if it were generated from $F_{Y|X=x_{0}}$. 
Second, given the i.i.d. assumption, the subsample $\{Y_{i,[x_{0}]}\}_{i=1}^{I}$ collected from all individuals serves as a random sample \textit{approximately} generated from $F_{Y|X=x_{0}}$. 
We can thus use the extreme order statistics of $\{Y_{i,[x_{0}]}\}_{i=1}^{I}$ to estimate the Pareto exponent $\alpha\left(x_{0}\right)$ by using one of the many existing methods.
In particular, if we apply the estimator of \citet{Hill1975} to this NN sample $\{Y_{i,[x_{0}]}\}_{i=1}^{I}$, then we get an estimator
\begin{equation}
\hat{\alpha}(x_{0})=\left[ \frac{1}{K}\sum_{k=1}^{K}\{\log
(Y_{(k),[x_{0}]})-\log \left( Y_{(K+1),[x_{0}]}\right) \}\right] ^{-1},
\label{eq:hill}
\end{equation}
of the conditional Pareto exponent $\alpha(x_0)$.
See Appendix \ref{sec:choice} for how to choose $K$ in practice.

We show in Appendix \ref{sec:primitive_conditions} that this estimator asymptotically follows the normal distribution as
\begin{equation}\label{eq:asymptotic1}
\sqrt{K}\left( \hat{\alpha}(x_{0})-\alpha (x_{0})\right) \overset{d}{\rightarrow }\mathcal{N}\left( 0,\alpha (x_{0})^{2}\right) 
\end{equation}
under suitable conditions.
Moreover, for any $x_{0}\neq x_{1}$, we also have the joint asymptotic normality:
\begin{equation}\label{eq:asymptotic2}
\sqrt{K}\binom{\hat{\alpha}(x_{0})-\alpha (x_{0})}{\hat{\alpha}(x_{1})-\alpha (x_{1})}\overset{d}{\rightarrow }\mathcal{N}\left( 0,\left( \begin{array}{cc}
\alpha (x_{0})^{2} & 0 \\ 
0 & \alpha (x_{1})^{2}
\end{array}
\right) \right) .
\end{equation}
We provide a formal statement of these results as Proposition \ref{prop:main} in Appendix \ref{sec:primitive_conditions}.

Given \eqref{eq:asymptotic1} and \eqref{eq:asymptotic2}, we can construct the standard t and F tests for our hypothesis testing problem (\ref{hypo}). 
Specifically, we reject the (one-sided) null hypothesis at the 5\% nominal level if 
\begin{equation}
\sqrt{K}{\;}\frac{\hat{\alpha}(x_{0})-r}{\hat{\alpha}(x_{0})}<\Phi^{-1}(0.05),  \label{test value}
\end{equation}
where $\Phi ^{-1}\left( \cdot \right) $ denotes the quantile function of the standard normal distribution and $r=2$, $3$, and $4$. 
Moreover, we reject the null hypothesis that $\alpha (x_{0})=\alpha (x_{1})$ (against the two-sided alternative) at the 5\% nominal level if
\begin{equation}
\sqrt{K}\frac{\left\vert \hat{\alpha}(x_{0})-\hat{\alpha}(x_{1})\right\vert }{\sqrt{\hat{\alpha}(x_{0})^{2}+\hat{\alpha}(x_{1})^{2}}}>\Phi^{-1}(1-0.025).  \label{test equal}
\end{equation}
Appendix \ref{sec:simulation} contains a simulation study, which supports our theoretical results.

We end this section with discussions about the related econometrics and statistics literature. 
To estimate the \textit{unconditional} Pareto exponent of some underlying distribution, researchers have developed numerous methods, including the popular and widely used \citeauthor{Hill1975}'s (\citeyear{Hill1975}) estimator,\footnote{This method fits large values in the sample to the Pareto distribution and implements the maximum likelihood estimation.} and those proposed by \citet{Smith1987MLE} and \citet{Gabaix2011rank}. 
See \citet{Fedotenkov2020} for a recent review of other methods. 
Primitive conditions and asymptotic properties of these estimators are typically studied by using extreme value theory and the Pareto tail approximation \citep[e.g.,][]{Hall1982,HallWelsh1985}. 
We refer to \citet{deHaan2007book} for a comprehensive review.
In recent work, \citet{beare2017determination} show that the cumulative sum of Markov multiplicative processes generates an approximate Pareto tail. This result supports our theoretical conditions and fits our empirical setup. 
In general, a pitfall of these methods in our framework is that the unconditional Pareto exponent cannot characterize the effect of important demographic characteristics on growth risk. 
This issue motivates us to study the conditional Pareto exponent. 

Unlike the unconditional Pareto exponent of, say $Y$, estimating its \textit{conditional} counterpart given some other variable $X$ has received less attention.
This is technically more challenging when $X$ contains continuous random variables as one could not obtain a random sample from the conditional distribution of interest, say $F_{Y|X=x_0}$ for some pre-specified value $x_0$.   
\citet{Wang2009} and \citet{WangLi2013} impose some parametric assumptions such that the Pareto exponent is a single index function of $X$.
\citet{Gardes2008} and \citet{Gardes2012} develop fully nonparametric methods based on local smoothing.
Estimation and inference based on these local smoothing methods sensitively rely on the bandwidth parameter unlike our proposed method based on nearest neighbors.
Furthermore, these existing methods underperform compared to ours in terms of bias and standard deviation, as demonstrated through simulation studies presented in Appendix \ref{sec:simulation}.
For these reasons, we suggest using our novel method proposed above.


Finally, in two related papers, \citet{trapani2016testing} and \citet{degiannakis2022} propose randomized testing procedures to test the finiteness of moments of a random variable or the residual of the linear regression.
Unlike their proposal, we test the existence of moments of the conditional distribution of a random variable given some other variables. 
To see the difference, consider, for example, the model that $Y = \mu(X) + \sigma(X)U$ where $\mu(X)$ and $\sigma(X)$ are some functions of $X$ that characterize the location and scale of $Y$ respectively.
One can estimate these two functions and further estimate the tail heaviness of $Y$ by that of the residual $\hat{U}$. 
But such a model by construction assumes that the tail heaviness of $Y$ remains unchanged across $X$ \citep*[cf.][]{WangLi2013}. 
This feature is neither sufficient for our purpose of measuring the tail heaviness of earnings growth as a function of demographic characteristics nor coherent with the empirical findings.\footnote{The tests developed by \citet{trapani2016testing} and \citet{degiannakis2022} can be generalized to conditional distributions. 
When $X$ is discrete, we can condition on specific values of $X$ by taking the subsample. 
Otherwise, we can select the subsample within a local neighborhood of a certain query value $X=x$. 
We leave this for future research.} 
Also related is the paper by \citet{sasaki2022fixed}, who propose inference for the tail index of conditional distributions.
Although allowing the tail index to depend on $X$, their inference method builds on the fixed-$k$ asymptotics and cannot be easily extended for estimation. 
In our context, estimation is valuable and thereby, we can propose a novel measure of earnings risk based on the estimate of the tail index unlike any of these existing papers. 
As such, our proposal is suitable for studying the tail features of the conditional distribution of income growth and relates these to the results in the income literature.   

\begin{remark}
We propose to use the estimate of the Pareto exponent of the conditional distribution as a novel measure of tail earnings risk given the conditioning variables, $X$. 
Our risk measure does `not' depend on $K$ (or any specific quantiles), but depends only on the underlying distribution (and the conditional value $x_0$), like the existing measures in the literature. 
In particular, the Pareto exponent is uniquely determined by the underlying distribution. See, for example, \citet{Gabaix2016} for some commonly used distributions and their corresponding Pareto exponents. 
In contrast, $K$ is a tuning parameter, which is close in spirit to the bandwidth in kernel regression. 
\end{remark}

\section{Results}\label{sec:empirical}

\subsection{Benchmark Sample}\label{sec:results_admin}

Applying the method introduced in Section \ref{sec:econometric_method} (and also in more detail in Appendix \ref{sec:additional_details}) to NESPD described in Section \ref{sec:admin}, we analyze the conditional tail risk of earnings of adult individuals in the United Kingdom.
We define $Y$ by the absolute difference between the log earnings in 2007 and the log earnings in 2008 for our baseline analysis.
For the conditioning variables $X$, we include the quantile of earnings level and the age of the individual in the base year (2007), following \citet{guvenen2021data}.
With this setting, we study the conditional Pareto exponent $\alpha(x_0)$ for each point $x_0$ of earnings levels from $\{0.05,0.10,\cdots,0.90,0.95\}$ (in quantile) and ages from $\{30,40,50\}$ for each of men and women.

Figure \ref{fig:nespd_2007} illustrates the estimates of the conditional Pareto exponents $\alpha(x_0)$ (in black lines) along with the upper bounds of their one-sided 95\% confidence intervals (in gray lines) for 30-, 40- and 50-year-old individuals. 
The left (respectively, right) panel shows the results for men (respectively, women) in each figure, which suggests different patterns of heterogeneous earnings risk across age, base-year earnings, and gender groups.

\begin{figure}
	\centering
		\includegraphics[width=0.49\textwidth]{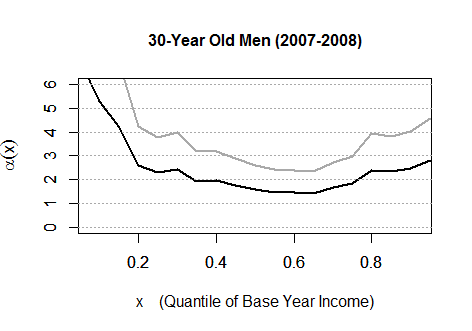}
		\includegraphics[width=0.49\textwidth]{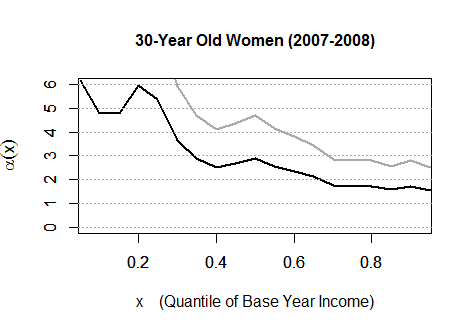}
		\includegraphics[width=0.49\textwidth]{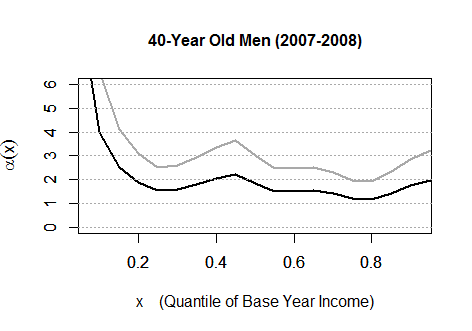}
		\includegraphics[width=0.49\textwidth]{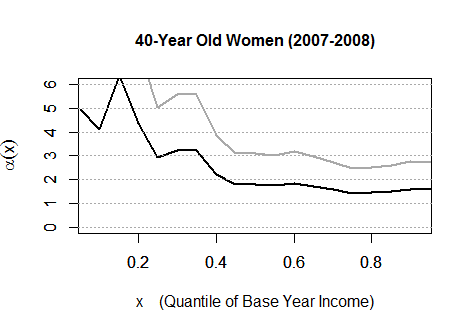}
		\includegraphics[width=0.49\textwidth]{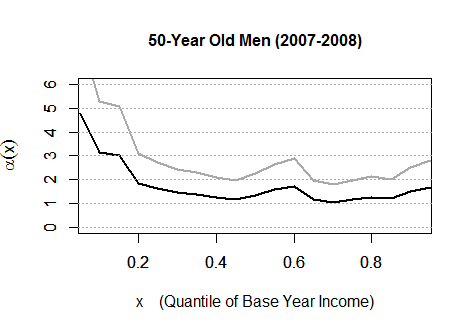}
		\includegraphics[width=0.49\textwidth]{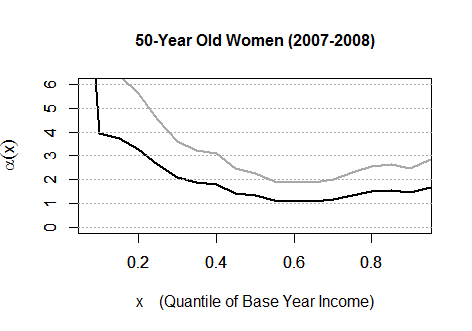}
	\caption{Estimates (black lines) and the one-sided 95\% confidence intervals (gray lines) of the Pareto exponents $\alpha(x_0)$ of the conditional tail risk for men (left) and women (right) based on the NESPD in the period 2007--2008. The left (respectively, right) column shows the results for men (respectively, women). The top, middle, and bottom panels show the results for 30-, 40- and 50-year-old individuals.}
	\label{fig:nespd_2007}
\end{figure}

%
%
%

We describe the empirical findings in the order of age and gender. 
For 30-year-old men (the top left panel in Figure \ref{fig:nespd_2007}), the conditional Pareto exponents (in point estimates) range from 1.4 to 7.0, and the upper bounds of the one-sided 95\% confidence intervals range from 2.3 to 11.4.
Given any quantile of earnings received in 2007, except the very top and bottom quantiles, the conditional Pareto exponent is significantly less than four, implying that the conditional kurtosis of earnings growth does not exist for these middle-earnings groups of young men.
Overall, the kurtosis barely exists for most of the base-year earnings levels even if we fail to reject the hypothesis of finite kurtosis.
Furthermore, given that earnings received in 2007 were between the 45th and the 75th percentile, the Pareto exponent is significantly less than three, implying that even the conditional skewness does not exist.
For 30-year-old women (the top right panel in Figure \ref{fig:nespd_2007}), the conditional Pareto exponents (in point estimates) range from 1.5 to 6.1, and the upper bounds of the one-sided 95\% confidence intervals range from 2.5 to 10.0.
For this subpopulation, the conditional Pareto exponent is significantly less than four for high- and middle-high-earnings groups of women, those above the 60th percentile of the earnings distribution in 2007, implying that the conditional kurtosis does not exist. 
Furthermore, given that earnings received in 2007 were at or above the 70th percentile, the Pareto exponent is significantly less than three, implying that even the conditional skewness does not exist.
Comparing the results between men and women at age 30, we observe that women were more vulnerable to earnings risk than men at the top quantiles of the base-year income level, above the 75th percentile.

For 40-year-old men (the middle left panel in Figure \ref{fig:nespd_2007}), the conditional Pareto exponents (in point estimates) range from 1.2 to 9.1, and the upper bounds of the one-sided 95\% confidence intervals range from 1.9 to 14.9.
Remarkably, the earnings risks of 40-year-old men are higher than those of 30-year-old men.
For this age group of men, we reject the hypothesis of finite kurtosis at any level of base-year earnings, apart from the bottom quantiles, below the 15th percentile.
Furthermore, given that earnings received in 2007 were between the 20th and the 35th percentile and above or at the median excluding the very top quantile (between the 75th and 80th percentile), the Pareto exponent is significantly less than three (two), implying that even the conditional skewness (respectively, standard deviation) does not exist.
For 40-year-old women (the middle right panel in Figure \ref{fig:nespd_2007}), the conditional Pareto exponents (in point estimates) range from 1.4 to 6.4, and the upper bounds of the one-sided 95\% confidence intervals range from 2.4 to 11.0.
For this age group of women, we reject the hypothesis of finite kurtosis at any earnings level above the 40th percentile.
Furthermore, given that earnings received in 2007 were at or above the 65th percentile, the Pareto exponent is significantly less than three, implying that even the conditional skewness does not exist.
Comparing the results between men and women at age 40, we observe that men are almost as vulnerable to earnings risk as women for this middle age group, with some difference at the bottom quantiles.

For 50-year-old men (the bottom left panel in Figure \ref{fig:nespd_2007}), the conditional Pareto exponents (in point estimates) range from 1.1 to 4.8, and the upper bounds of the one-sided 95\% confidence intervals range from 1.8 to 8.0.
For this age group of men, we reject the hypothesis of finite kurtosis and finite skewness at any level of base-year earnings, apart from the bottom quantiles, below the 20th percentile. 
Moreover, we reject the hypothesis of finite standard deviation at the 45th percentile and between the 65th and the 75th percentile.
For 50-year-old women (the bottom right panel in Figure \ref{fig:nespd_2007}), the estimated conditional Pareto exponents range from 1.1 to 15.0, and the upper bounds of the one-sided 95\% confidence intervals range from 1.9 to 25.7.
For this age group of women, we reject the hypothesis of finite kurtosis at any level of base-year earnings, apart from the bottom quantiles, below the 30th percentile. 
Moreover,  we reject the hypothesis of finite skewness (standard deviation) at any level of base-year earnings above the 40th percentile (between the 55th and the 70th percentile).

The tuning parameter $K$ is to be selected by a data-driven method (cf. Appendix \ref{sec:choice}) and should not be manipulated by a researcher. Yet, it may be of interest to see the sensitivity of the results to variations in $K$. Table \ref{tab:conditional_robK_50_07} reports the estimates of the conditional Pareto exponent, $\alpha \left(x_0\right)$, for several values of $K$ for 50-year old men in 2007--2008, conditional on the quantiles of base year income. Our findings are robust across different values of $K$: estimates tend to be higher at lower quantiles of the base year earnings distribution, compared to those at the higher end.\footnote{We obtain similar findings in terms of robustness to $K$ of the Pareto estimates obtained in the following sections for other years or values of the conditioning variables. Results are available from the authors upon request.}
\begin{table}
	\centering
	\begin{tabular}{rcccccc}
		\hline\hline
		\multicolumn{7}{c}{50-Year Old Men (2007--2008)} \\
		\multicolumn{1}{l}{$N=197$} & \multicolumn{6}{c}{$K$}\\
		\cline{2-7}
		\multicolumn{1}{c}{Percentile Lag $Y$ } &$K^* = 7$ & 0.10N & 0.05N & 0.04N & 0.03N  & 0.02N  \\\hline
		1 &3.864&2.820&2.839&3.864&3.364&2.797\\
		2&3.204&2.520&3.568&3.204&4.801&3.707\\
		3&3.367&1.428&2.232&3.367&3.414&4.200\\
		4&1.373&1.105&1.556&1.373&1.076&1.057\\
		5&1.089&1.065&1.065&1.089&1.219&1.257\\
		6&1.171&1.532&1.410&1.171&1.641&1.984\\
		7&0.955&1.148&0.953&0.955&1.385&0.998\\
		8&0.954&1.086&0.835&0.954&1.662&2.168\\
		9&1.073&1.109&0.857&1.073&2.257&1.817\\
		10&1.837&1.305&1.413&1.837&1.578&0.997\\
		11&3.201&1.739&2.130&3.201&2.659&8.015\\
		12&1.419&1.962&1.531&1.419&1.424&6.588\\
		13&0.936&1.535&1.009&0.936&0.713&0.899\\
		14&0.662&1.197&0.744&0.662&0.630&1.946\\
		15&0.786&1.222&0.908&0.786&0.782&1.945\\
		16&0.967&1.143&1.006&0.967&0.824&2.487\\
		17&1.271&1.221&1.083&1.271&1.023&2.534\\
		18&1.047&1.330&1.081&1.047&1.216&2.337\\
		19&1.330&1.903&1.541&1.330&1.515&1.243\\
		\hline
		\hline\hline
	\end{tabular}
	\caption{Sensitivity to $K$ of the estimates of the Pareto exponent for 50-year old men in 2007-2008 conditional on the quantiles of base year income.}
	\label{tab:conditional_robK_50_07}
\end{table}

%
%

\subsection{Period of Positive Growth}\label{sec:robustness_expansion}

While our main focus has been on the period of the great recession between 2007 and 2008, we next look at the period of positive growth, between 2015 and 2016.
Figure \ref{fig:nespd_2015} illustrates the results for earnings growth between 2015 and 2016 for 30-, 40- and 50-year-old individuals. 
These results share similar qualitative patterns to those reported in Figure \ref{fig:nespd_2007}. 
First, 30-year-old men at high quantiles of base-year earnings enjoy less earnings risk.
Second, 40-year-old men have an overall higher earnings risk than 30-year-old men.
Lastly, and most remarkably, for men the earnings risk is not necessarily lower in the period 2015--2016 than in the period 2007--2008, even though the former period enjoyed positive GDP growth (2.2\% in 2015 and 1.8\% in 2016) and the latter period suffered from negative GDP growth (0.7\% in 2007 and $-$4.4\% in 2008).
Women, instead, are less vulnerable to earnings risk in the period 2015--2016 than in the period 2007--2008.
Overall, there are more similarities than differences despite the contrast between a recession and a positive growth in the UK economy.

\begin{figure}
	\centering
		\includegraphics[width=0.49\textwidth]{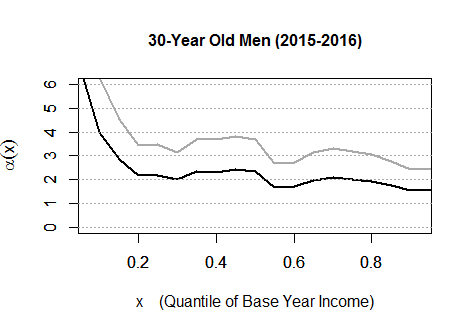}
		\includegraphics[width=0.49\textwidth]{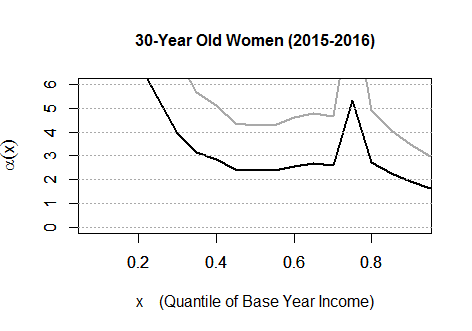}
		\includegraphics[width=0.49\textwidth]{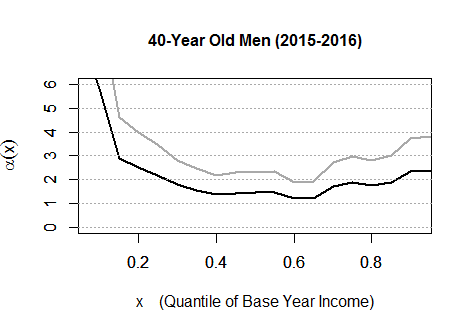}
		\includegraphics[width=0.49\textwidth]{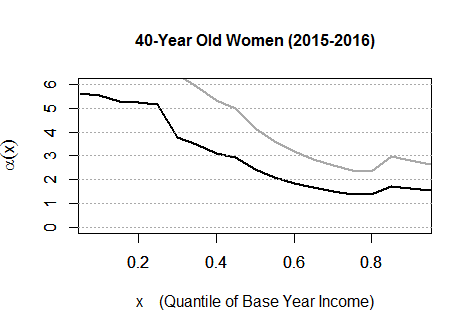}
		\includegraphics[width=0.49\textwidth]{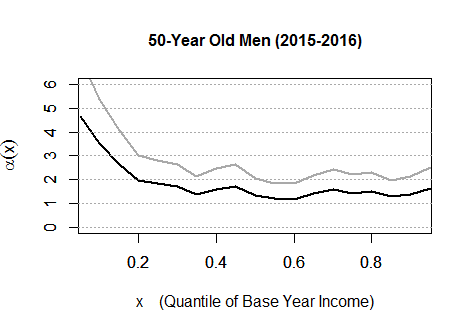}
		\includegraphics[width=0.49\textwidth]{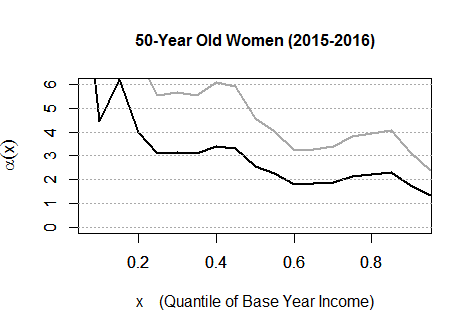}
	\caption{Estimates (black lines) and the one-sided 95\% confidence intervals (gray lines) of the Pareto exponents $\alpha(x_0)$ of the conditional tail risk for men (left) and women (right) based on the NESPD in the period 2015--2016. The left (respectively, right) column shows the results for men (respectively, women). The top, middle, and bottom panels show results for 30-, 40- and 50-year-old individuals.}
	\label{fig:nespd_2015}
\end{figure}

In summary, we find that:
1) the kurtosis, skewness, and even standard deviation may not exist for the conditional distribution of earnings growth given certain attributes (age, gender, and earnings);
2) 40- and 50-year-old workers have overall higher earnings risk than 30-year-old workers and we thus document that tail earnings risk is increasing over the life cycle;  and
3) these patterns appear both in the period 2007--2008 of great recession and the period 2015--2016 of positive growth for men, while there are some differences for women.
Finally, we remark that, although we use the two periods, 2007--2008 and 2015--2016, to draw the above conclusion, we also present additional empirical results for other periods between 2005 and 2016 in Appendix \ref{sec:additional_results} to demonstrate robustness. 
In Appendix \ref{sec:psid}, we also apply our proposed method to the US PSID and obtain similar empirical findings on the tail heaviness of the conditional distribution of income risk. 


\subsection{Job Stayers}\label{sec:robustness_stayers}

We also experiment with a different set of selection criteria: we restrict the sample to workers who did not change their jobs in the last twelve months. We refer to them as job stayers. 
Figures \ref{fig:nespd_2007_samejob} and \ref{fig:nespd_2015_samejob} illustrate estimates of the conditional Pareto exponents $\alpha(x_0)$ (in black lines) along with the upper bounds of their one-sided 95\% confidence intervals (in gray lines) for the two pairs of years: 2007--2008 and 2015--2016, with the restricted sample of job stayers. 
Figures \ref{fig:nespd_2007_samejob} and \ref{fig:nespd_2015_samejob} show that, for almost all attributes, conditional kurtosis does not exist at any level of baseline earnings. 

This finding is in line with the US evidence documented by \citet{guvenen2021data} that job stayers experience different earnings dynamics and, in particular, face earnings innovations that are more leptokurtic, especially at the bottom of the earnings distribution. 

More specifically, we find that, at low earnings levels, conditional kurtosis does not exist, in contrast to what we observe for the whole population. Excluding middle-aged men at top quantiles of earnings, we reject the hypothesis of finite conditional kurtosis for job-stayers at all earnings levels.

\begin{figure}
	\centering
		\includegraphics[width=0.49\textwidth]{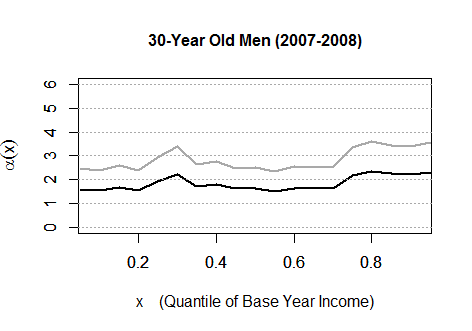}
		\includegraphics[width=0.49\textwidth]{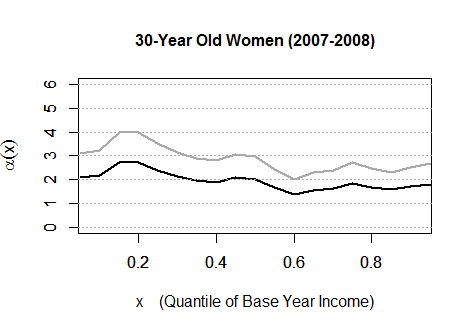}
		\includegraphics[width=0.49\textwidth]{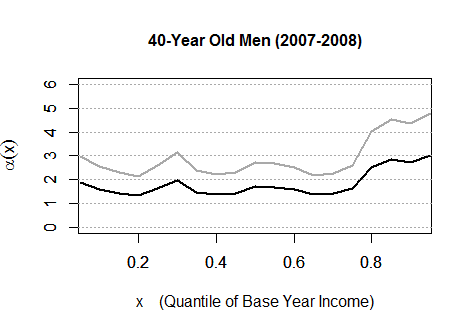}
		\includegraphics[width=0.49\textwidth]{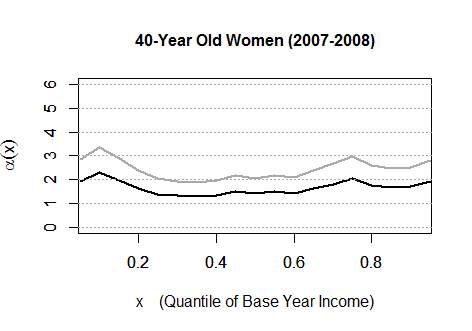}
		\includegraphics[width=0.49\textwidth]{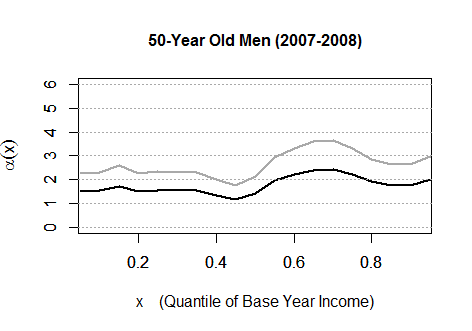}
		\includegraphics[width=0.49\textwidth]{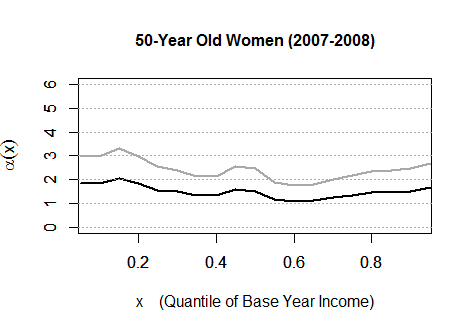}
	\caption{Estimates (black lines) and the one-sided 95\% confidence intervals (gray lines) of the Pareto exponents $\alpha(x_0)$ of the conditional tail risk for men (left) and women (right) based on the NESPD in the period 2007--2008, for job-stayers. The left (respectively, right) column shows the results for men (respectively, women). The top, middle, and bottom panels show results for 30-, 40- and 50-year-old individuals. Number of individuals: 32,411 men (32,509 women).}
	\label{fig:nespd_2007_samejob}
\end{figure}

\begin{figure}
	\centering
		\includegraphics[width=0.49\textwidth]{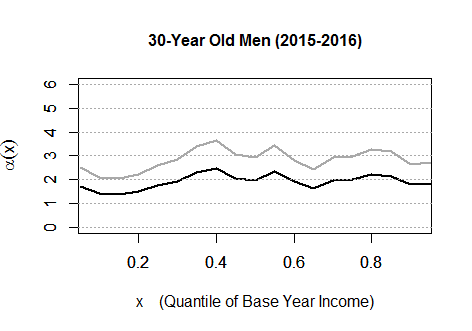}
		\includegraphics[width=0.49\textwidth]{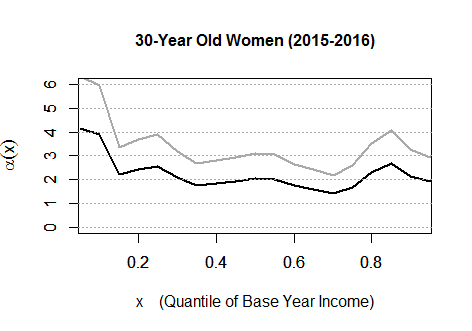}
		\includegraphics[width=0.49\textwidth]{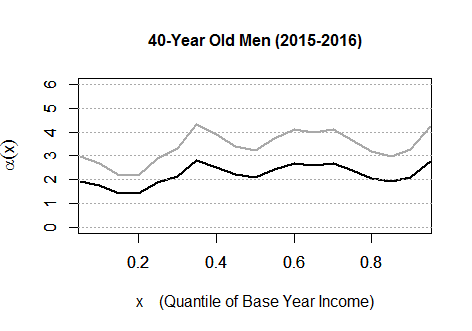}
		\includegraphics[width=0.49\textwidth]{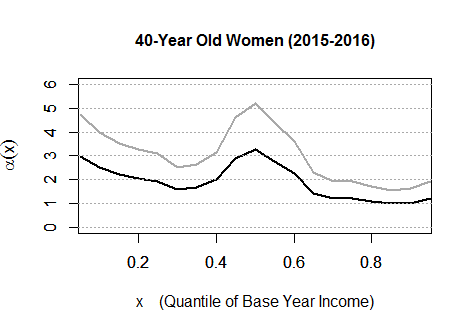}
		\includegraphics[width=0.49\textwidth]{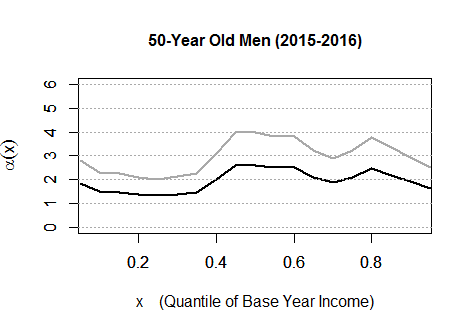}
		\includegraphics[width=0.49\textwidth]{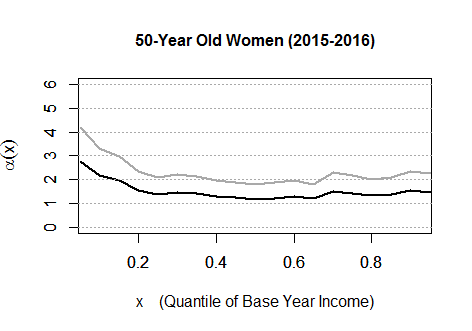}
	\caption{Estimates (black lines) and the one-sided 95\% confidence intervals (gray lines) of the Pareto exponents $\alpha(x_0)$ of the conditional tail risk for men (left) and women (right) based on the NESPD in the period 2015--2016, for job-stayers. The left (respectively, right) column shows the results for men (respectively, women). The top, middle, and bottom panels show results for 30-, 40- and 50-year-old individuals. Number of individuals: 38,680 men (41,393 women).}
	\label{fig:nespd_2015_samejob}
\end{figure}

\subsection{Pareto Tail Fit}\label{sec:Pareto}
In this section, we would like to comment more on Assumption 2 by providing further information about the upper tail of the conditional distribution of earnings changes.
To this goal, we present the plots of quantiles in Figure \ref{fig:Paretoplot_cond}. These plots illustrate the goodness of the Pareto fit to the data at the tail. In particular, we plot the ratio $(Y_{(1)}/Y_{(K)}, ...,Y_{(K-1)}/Y_{(K)})$, with $K = 0.3N$ for the empirical quantiles (solid blue lines) against the Pareto-fitted ones (dashed black lines), obtained as $t^{-1/{\hat{\alpha}(x_0)}}$ for $t \in [0,1]$. 
Here $Y_{(i)}$ denotes the $i$-th largest observation among all $Y_i$'s in the subsample and $\hat{\alpha}(x_0)$ denotes the classic Hill's estimator. 
The figure shows a very good Pareto fit at the tail, providing suggestive evidence that the tail follows a power law behavior. 
The evidence in the data also suggests that the largest observations are far away from each other and do not converge to any finite constant. These features imply that a truncation or a finite upper bound is not coherent with our data set. There are sufficient reasons to care about extreme events in untruncated distributions in this application. 

\begin{figure}
	\centering
	\begin{tabular}{@{}c@{}c@{}c@{}}
		\multicolumn{3}{c}{Median Base Year Income - Men (2007--2008)} \\
				& & \\
		30-Year Old & 40-Year Old& 50-Year Old \\
		\includegraphics[width=0.3\textwidth]{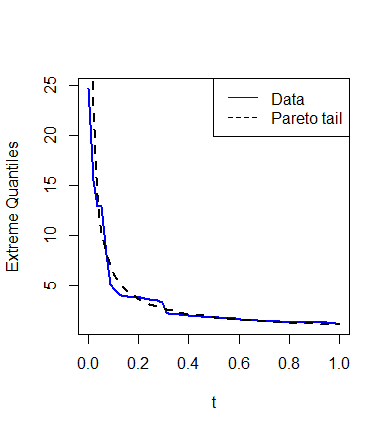}&
		\includegraphics[width=0.3\textwidth]{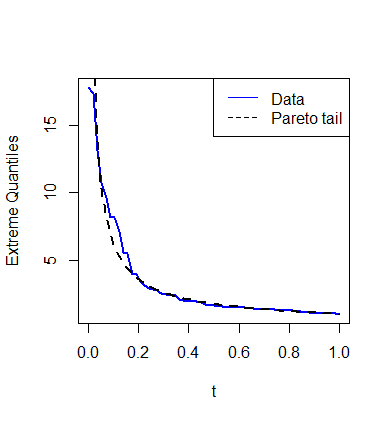}&
		\includegraphics[width=0.3\textwidth]{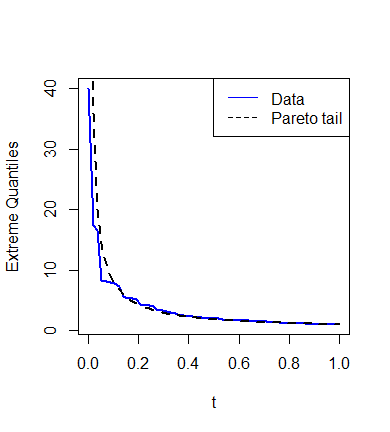} \\
	\end{tabular}
	\caption{Pareto plots for 30-, 40-, 50-year old men in 2007-2008, conditional on being in the median of the distribution of base year earnings, when using the largest $K=0.3N$ tail observations.}
	\label{fig:Paretoplot_cond} 
\end{figure}

\subsection{Unconditional Tail Risk}\label{sec:unconditional}
Thus far, our analyses have focused on conditional distributions of earnings changes following the literature.
In this section, we conduct the analysis for the unconditional distribution of earnings changes. Table \ref{tab:unconditional_estimates} reports the estimates of the Pareto exponents $\alpha$ for men and women, respectively, for the years 2007--2008 and 2015--2016. Table \ref{tab:unconditional_estimates} further shows the sensitivity of the results to different choices of the order statistics $K$ used to build the Hill's estimator. 

\begin{table}
	\centering
	\begin{tabular}{rcccc}
	\hline\hline
	\multicolumn{5}{c}{Men (2007--2008)} \\
	\multicolumn{1}{l}{$N=38955$} & \multicolumn{4}{c}{$K$}\\
	\cline{2-5}
	\multicolumn{1}{l}{ } & 0.04N & 0.025N & $K^* = 780$ & 0.02N  \\\hline
	alpha          &  2.280 & 3.129 & 3.714 & 3.712  \\
	SD             & .058 & .100 & .133 & .133  \\
	\hline
	\multicolumn{5}{c}{Women (2007--2008)} \\
	\multicolumn{1}{l}{$N=39576$} & \multicolumn{4}{c}{$K$}\\
	\cline{2-5}
	\multicolumn{1}{l}{ } & 0.04N & 0.025N & $K^* = 792$ & 0.02N  \\\hline
	alpha          & 2.446 & 3.158 & 3.549 & 3.561 \\
	SD             & .061 & .100 & .126 & .127  \\
	\hline
	\multicolumn{5}{c}{Men (2015--2016)} \\
	\multicolumn{1}{l}{$N=45985$} & \multicolumn{4}{c}{$K$}\\
	\cline{2-5}
	\multicolumn{1}{l}{ } & 0.04N & 0.025N & $K^* = 920$ & 0.02N   \\\hline
	alpha          &  2.259 & 2.898 & 3.516 & 3.526  \\
	SD             &  .053 & .085 & .116 & .116  \\
	\hline
	\multicolumn{5}{c}{Women (2015--2016)} \\
	\multicolumn{1}{l}{$N=49921$} & \multicolumn{4}{c}{$K$}\\
	\cline{2-5}
	\multicolumn{1}{l}{ } & 0.04N & 0.025N & $K^* = 999$ & 0.02N   \\\hline
	alpha          &   2.567 & 3.272 & 3.688 & 3.702  \\
	SD             &  .057 & .093 &.117 & .117  \\
	\hline\hline
\end{tabular}
	\caption{Estimates of the Pareto exponent with their standard errors for the unconditional distribution of earnings changes for men and women in 2007-2008 and 2015-2016.}
	\label{tab:unconditional_estimates}
\end{table}

Furthermore, we provide more information about the upper tail of the unconditional distribution of earnings changes in our data.  
We report the Pareto plots in Figure \ref{fig:pareto_plot_K30} to illustrate the goodness of the Pareto fit to the data at the tail. 
Specifically, we plot the ratio $(Y_{(1)}/Y_{(K)}, ...,Y_{(K-1)}/Y_{(K)}$, with $K = 0.3N$ for the empirical quantiles (solid blue line) against the Pareto fitted ones (dashed black line), obtained as $t^{-1/{\hat{\alpha}}}$ for $t \in [0,1]$. 
Here $Y_{(i)}$ denotes the $i$-th largest observations among all $Y_i$'s and $\hat{\alpha}$ again denotes the Hill estimator using these $Y_i$'s. 
Figure \ref{fig:pareto_plot_K30} shows the Pareto tail fit for the sample of men (left panel) and for two other restricted samples of men: job stayers as defined in Section \ref{sec:robustness_stayers} and those with non-missing observations from $t-3$ to $t+1$ men, with $t$ being 2007.

This evidence is in line with the log-log density of one-year and five-year earnings changes provided in the supplemental material of \citet{bell2021time}. Specifically, they estimate the slope for the left and right tails respectively for the region ($-$4,$-$1) and (1,4), which can be interpreted as the estimates of the Pareto exponent. 
%

\begin{figure}
	\centering
	\includegraphics[width=0.3\textwidth]{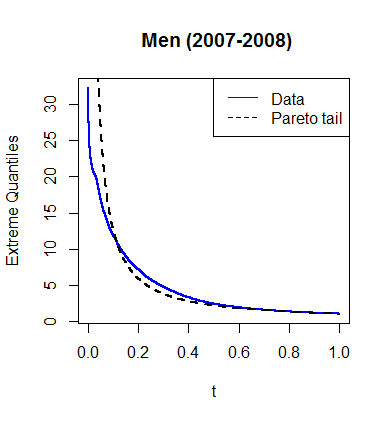}
		\includegraphics[width=0.3\textwidth]{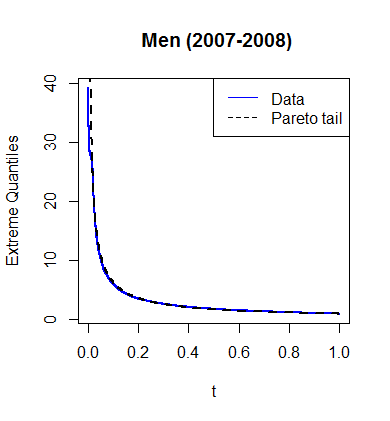}
		\includegraphics[width=0.3\textwidth]{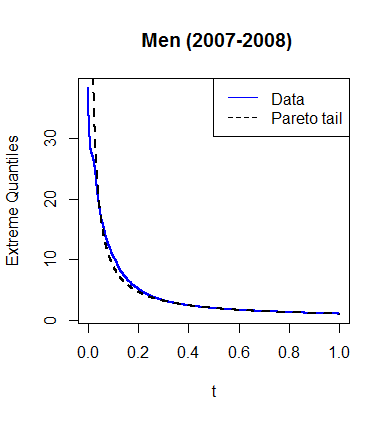}
	\caption{Pareto plots for the unconditional distribution of earnings changes for the benchmark sample (first panel), for Job-Stayers (second panel) and for those with non-missing observations over the period 2004--2008 (third panel), men, for 2007--2008, when using the largest $K=0.3N$ tail observations.}
	\label{fig:pareto_plot_K30}
\end{figure}

\subsection{Discussion}\label{sec:discussion}

Our findings confirm that moment-based measures, such as variance, skewness, or kurtosis, may not be well-defined for some sub-populations. This result points to the importance of using alternatives to the moment-based measures of earnings risk typically used in the literature. We propose a measure of conditional tail earnings risk that overcomes this limitation and offers new insights into the amount of tail risk. Our measure quantifies the amount of risk related to, possibly extreme, shocks individuals experience, given their lagged income and age profile. 

Our first finding is that, compared to younger individuals, older individuals are more likely to be hit by potentially large shocks, e.g., health shocks or occupation changes. For the US, \citet{guvenen2021data} document that earnings growth of males aged 45-55 and earning \$100,000 has a kurtosis of 18, compared to 5 for younger workers earning \$10,000 (and 3 for a Gaussian distribution). These findings highlight the importance of life cycle effects in tail earnings risk, but these estimates are unreliable since kurtosis might not exist in the first place. Our methodology allows measuring this conditional tail risk and indeed reveals that older workers are much more vulnerable to extreme shocks (estimates of our measure are close to 1 along a large part of the distribution of lagged earnings) compared to younger individuals (for which our estimates are close to 2 or larger). Similarly, in our second finding, we show that job stayers face earnings shocks that are more leptokurtic, especially at the bottom of the earnings distribution (at the bottom of the earnings distribution, estimates of our measure go from 5 for the whole sample to around 1 for job stayers). In addition, the frequency and occurrence of leptokurtic shocks do not seem to differ in recessions. A further investigation of the nature of these shocks might shed more light in this direction.
To sum up, our results for the whole sample demonstrate that tail earnings risk increases over the life cycle and is higher for job stayers than for movers. Moreover, at least for the male subsample, there are no relevant differences between recession and expansion periods in our proposed measure of tail earnings risk. 

It is worth discussing how our methodology and results differ from those of other papers in the literature on earnings risks. 
\citet{guvenen2021data}, among others, use moment-based measures, which we have shown to be not well-defined, and complement the analysis with quantile-based measures. Although well-defined when moments may fail to exist, quantile-based measures do not capture extreme events since they do not include observations larger or smaller than some given percentiles. Given our interest in conditional tail earnings risk, the excluded tail observations contain valuable information for the analysis. Among other things, information on tail events is central to quantifying earnings risk under heavy-tailed distributions. Thus, the quantile-based approach may not be informative about extreme income changes and tail events.

Another measure proposed in the literature is the coefficient of variation (CV) of \citet{arellano2021}, computed as the ratio between the mean absolute deviation of income divided by the mean income, conditional on a set of variables. Estimating the numerator and denominator of the CV are two prediction tasks performed with Poisson regressions with both micro and macro predictors. The main idea is to capture the agent's prediction problem by focusing on features of the predictive distribution of income.\footnote{We refer to \citet{arellano2021} for a detailed description of their methodology and the set of predictors used in the analysis. Note that the CV remains well-defined with zero income.} One of their main findings is that younger individuals face higher levels of income risk, and the dispersion of the CV markedly increases in recessions. These results contrast what we obtain with our measure of conditional tail risk. In order to understand what drives the differences between our results and their findings, we build the CV measure of earnings risk with NESPD for the UK to separate the country and institutional effect from the differences due to the different methodology and object of interest. Indeed, the divergences could reflect differences in terms of institutional context, sample selection, and the types of measures used to quantify income risk. There are several constraints in replicating their approach due to some data limitations. For instance, in our case, earnings are the only source of income as we do not observe unemployment benefits. We thus capture a less broad measure of risk compared to theirs.\footnote{When running the regression with both micro and macro variables, we cannot include the highest level of education achieved, unemployment benefits, or the number of days worked in previous years since these variables are not available in our data set.}  
Acknowledging these data limitations, we computed the CV for UK workers.\footnote{Results for the replication of \citet*{arellano2021} for the UK are available upon request.} We get qualitatively similar findings to those in \citet{arellano2021}, although the magnitude of the coefficient of variation is different. In particular, we get a smaller dispersion of the CV compared to what they report for Spain. As in their case, however, we obtain that younger individuals face higher risk: for this subgroup, the coefficient of variation is higher, with the 90th percentile of their CV being much larger than the corresponding CV percentile for other age groups. Thus, institutional settings, different sample sizes, and data sets matter for the magnitude of the results, not for the qualitative differences. We can conclude that the main differences between our findings and theirs come from the different concepts of risk we are measuring rather than differences in context (Spain vs. the UK in this particular case). As with quantile-based measures in \citet{guvenen2021data}, also with the CV measure proposed by \citet{arellano2021}, one might fail to capture extreme income changes and tail events. Our measure captures a different aspect of earnings risk: it quantifies tail risk, and thus, it should be seen as complementary to their proposal. 
 
We further conduct the following robustness check, focusing on earnings changes with the benchmark sample of men in 2007--2008. Following the approach used in the GRID project, we condition on the average of the three previous earnings, going from $t-1$ backward, rather than on base year earnings, i.e., earnings at time $t$. Figure \ref{fig:nespd_2007_avg3}, in Appendix \ref{app:avg}, plots the estimates, along with the one-sided 95\% confidence intervals, of the Pareto exponents $\alpha(x_0)$ of the conditional tail risk for men (left) and women (right) based on the NESPD in the period 2007--2008. This figure illustrates the robustness of our main findings: the increase of tail earning risk over the life cycle.

Finally, in our baseline analysis, we define $Y$ as the absolute value of log earnings changes. We repeat the analysis separately for the left and the right tail of the distribution of earnings changes.\footnote{Results for the right and left tails are available upon request.} For the left and right tails separately considered, the estimates of the Pareto exponent are always not significantly smaller than 4, and point estimates can get much higher values (up to 50 for the right tail, for some subpopulations). 
Overall, the patterns of the left and right tails resemble those of the absolute changes. However, despite these similarities, there are some remarkable differences, especially at the bottom and top of the earnings distribution. More specifically, for the left tail, estimates tend to be lower (taking a maximum value of around 11) and relatively constant or decreasing over the earnings quantiles. Instead, for the right tail, estimates are much more variable and lower at the bottom of the earnings distribution. This finding is consistent with the patterns of asymmetry of conditional Kelly's skewness documented in Figure \ref{fig:nespd_moments}: higher earners are more likely to face large negative earnings shocks, because of limited opportunities for large gains and higher risk of significant declines, while the vice versa holds for workers at the bottom of the earnings distribution.

	{\bf Implications:}
	We close this section with a discussion of some alternative methods that are robust to the non-existence of moments. 
First, if our risk measure implies that the $p$-th moment may not exist, then researchers should select only those calibration/estimation methods that exploit only up to the $(p-1)$-th moments.
For instance, many estimation/calibration methods proposed in the literature, including \citet*{daly2022improving} as well as others like \citet*{meghir2004income}, require only up to second moments of earnings/income changes to exist. 

Second, there are other estimation strategies that are not based on the method of moments. 
For instance, the estimation strategy proposed by \citet*{arellano2017earnings} is theoretically robust to the non-existence of moments. 
Although in their current implementation, they implicitly impose that moments exist, this is not required and, thus, a careful choice of some elements of their implementation, e.g., changing the Laplace tail assumption of the transitory shock, would suffice to make their proposal robust to the existence of moments.   
\citet*{de2020nonlinear} require that the conditional probability distributions over the bins they define exist, and \citet*{janssens2022finite} only impose some assumptions of continuity and the use of a sufficient number of grid points. 
Improvements in the discretization procedure might be obtained by careful tailoring of the approximation strategy to the tails (for instance, adapting \citet*{tauchen1986finite} to non-normal errors for the tail behavior) or by applying methods that rely on moments matching (as in \citet*{farmer2017discretizing} and \citet*{lkhagvasuren2023finite}) but restricting them to the existing moments.\footnote{The method proposed by Rouwenhorst (1995) relies on matching the first two moments only and is applicable to Gaussian AR(1) processes. Given the heavy tail feature, it is unappealing for earnings processes.}

\section{Concluding Remarks\label{sec:conclusion}}

The literature often relies on moment-based measures of earnings risk, such as the variance, skewness, and kurtosis. However, such moments may not exist in the population under heavy-tailed distributions.
In this paper, we show that the population kurtosis, skewness, and even variance indeed often fail to exist for the conditional distribution of earnings growths given age, gender, and past earnings.
Hence, moment-based analyses in the literature may not make sense under heavy-tailed distributions.

Heavy tails of income and earnings risk distributions are costly in economies with risk-averse agents.
Despite this common knowledge, the tail heaviness has been arguably less investigated in the literature on earnings and income dynamics. 
In light of the limited capacity of moment-based measures in quantifying the tail-heaviness, we consider the conditional Pareto exponent as a novel robust measure of conditional earnings risk given observed attributes of individuals and propose a method of estimation and inference about this measure. 

Applying the proposed method to the UK NESPD and, in the Appendix, to the US PSID, we obtain the following findings. 
First, as emphasized above, the population kurtosis, skewness, and even standard deviation might be infinite for the conditional distribution of income growth given certain attributes, such as age, gender, and lagged income, and hence, their sample counterparts might be less informative. 
Second, 40- and 50-year-old workers have overall higher earnings risk than 30-year-old workers, thus earnings risk increases over the life cycle.
Third, conditional earnings risk is higher for job stayers, in particular, at the bottom of the earnings distribution. Fourth, these patterns appear both in the periods of great recession and positive growth, while there are differences as well, especially for women.

To the best of our knowledge, this is the first work to shed light on tail features of heavy-tailed distributions of earnings and income growth with a formal econometric method.
While we focused on two specific data sets for empirical analysis, we hope that our proposed method will spur further empirical research with other data sets and for the study of other economies.

\newpage
\appendix
\section*{Appendix}

Appendix \ref{sec:additional_details} provides additional details about the econometric method introduced in Section \ref{sec:econometric_method}. 
Appendix \ref{sec:additional_results} provides additional empirical results with the UK NESPD.  
Appendix \ref{sec:psid} provides the results with the US PSID.

\section{Additional Details about Section \ref{sec:econometric_method}}\label{sec:additional_details}

Appendix \ref{sec:primitive_conditions} presents a formal statement of the asymptotic normality results, \eqref{eq:asymptotic1} and \eqref{eq:asymptotic2}.
Appendix \ref{sec:simulation} presents simulation studies. 
Appendix \ref{sec:choice} suggests a choice of $I$, $J$, and $K$ and a finite sample adjustment in practice. 
Appendix \ref{sec:adjustments} provides a refinement of our methods with a fixed-$K$ asymptotic analysis.

\subsection{Econometric Theory}

\label{sec:primitive_conditions}

In this appendix section, we provide a formal statement of the asymptotic normality results, \eqref{eq:asymptotic1} and \eqref{eq:asymptotic2}.
We consider the following conditions to establish these results.

\begin{condition}
\label{cond:1}

\begin{enumerate}

\item $(Y_{ij},X_{ij}^{\intercal })^{\intercal}$ is i.i.d. across $i$ and $j$.
In addition, $f_{X}\left( x\right) $ is uniformly
continuously differentiable and bounded away from zero in an open ball
centered at $x_{0}$.

\item $1-F_{Y|X=x}(y)=c(x)y^{-\alpha (x)}\left( 1+d(x)y^{-\gamma
(x)}+r(x,y)\right) $ uniformly as $y\rightarrow \infty $ where $c(\cdot )>0$
and $d(\cdot )$ are continuously differentiable functions and uniformly
bounded between $0$ and $\infty $, $\alpha (\cdot )>0$ and $\gamma \left(
\cdot \right) >0$ are continuously differentiable functions and $r(x,y)$ is
continuously differentiable with bounded derivatives w.r.t. both $x$ and $y$
and satisfies $\lim \sup_{y\rightarrow \infty }\sup_{x\in B_{\delta
}(x_{0})}\left\vert r(x,y)/y^{-\gamma (x)}\right\vert =0$ for some open ball 
$B_{\delta }(x_{0})$ centered at $x_{0}$ with radius $\delta >0$.

\item $K\rightarrow\infty$ and $K=o\left( I^{2\gamma (x_{0})/\left( \alpha (x_{0})+2\gamma
(x_{0})\right) }\right) \rightarrow 0$ as $n\rightarrow \infty $.

\item $I\rightarrow \infty $ and $\lim_{I\rightarrow \infty }J/I^{\lambda
}=C\in \left( 0,\infty \right) $ for some $C$ and $\lambda \in \left(
0,\infty \right) $. In addition, $\lambda /2\geq \gamma (x_{0})/(\alpha
(x_{0})+2\gamma (x_{0}))$.
\end{enumerate}
\end{condition}

Condition 1.1 requires the data to be independent across $i$ and $j$. 
Although we focus on the cross-sectional data set, estimation of tail index with time series data set has also been studied in the literature. 
See, for example, \citet{Hill2015} and reference therein. 
The asymptotic variance in general depends on the time series dependence, and one may apply some long-run variance estimator \citep[e.g.,][Theorem 3]{Hill2010}.
Regarding testing for infinite moments, the tests proposed by \citet{trapani2016testing} and \citet{degiannakis2022} also allow for time series dependence.

Condition 1.2 assumes that the target conditional distribution satisfies a second-order Pareto tail approximation, which is sufficient for Assumption 2. 
The parameter $\gamma (x)$ governs the preciseness of the Pareto tail approximation, and the approximation error $r(x,y)$ is asymptotically negligible. 
The unconditional version (i.e., without $x$) of this condition has been commonly imposed in the statistics literature. 
Essentially, this Pareto-type tail condition is satisfied by all distributions whose densities are smooth and monotonically decay to zero, as referred to the \textit{von Mises'} condition. 
See \citet{deHaan2007book} for a comprehensive review.

Condition 1.3 specifies the choice of $K$, a tuning parameter that characterizes the number of larger order statistics to be considered as stemming from the tail. 
Such a choice is to eliminate the asymptotic bias. 
If we choose $K$ such that $K\times I^{-2\gamma (x_{0})/\left( \alpha(x_{0})+2\gamma (x_{0})\right) }\rightarrow \mu (x_{0})$ for some constant $\mu (x_{0})\in \mathbb{R} $, then the asymptotic distribution in Proposition \ref{prop:main} involves one additional item $-\mu (x_{0})\xi (x_{0})$. 
To avoid such bias, we can select a smaller order $K$ as in Condition 1.3 such that $\mu (x_{0})$ becomes zero asymptotically. 
This is close in spirit to choosing the undersmoothing bandwidth in the standard kernel estimation.
Condition 1.4 requires a large $I$ and a large $J$. Note that $J$ can be substantially smaller than $I$ when $\lambda $ is less than one.

The following proposition characterizes the asymptotic behavior of $\hat{\alpha}\left( x_{0}\right) $ given in \eqref{eq:hill} under Condition \ref{cond:1}. 
\begin{proposition}
\label{prop:main}Suppose that Condition \ref{cond:1} holds. Then 
\begin{equation*}
\sqrt{K}\left( \hat{\alpha}(x_{0})-\alpha (x_{0})\right) \overset{d}{%
\rightarrow }\mathcal{N}\left( 0,\alpha (x_{0})^{2}\right) .
\end{equation*}%
Moreover, for any $x_{0}\neq x_{1}$, $\hat{\alpha}(x_{0})$ and $\hat{\alpha}%
(x_{1})$ are asymptotically independent, or equivalently,%
\begin{equation*}
\sqrt{K}\binom{\hat{\alpha}(x_{0})-\alpha (x_{0})}{\hat{\alpha}%
(x_{1})-\alpha (x_{1})}\overset{d}{\rightarrow }\mathcal{N}\left( 0,\left( 
\begin{array}{cc}
\alpha (x_{0})^{2} & 0 \\ 
0 & \alpha (x_{1})^{2}%
\end{array}%
\right) \right) .
\end{equation*}
\end{proposition}

\begin{proof}
For convenience of writing, we also introduce the notation $\xi (x_{0})=1/\alpha (x_{0})$,
which is referred to as the tail index. 
We estimate $\hat{\xi}\left(x_{0}\right) $ by the reciprocal of $\hat{\alpha}(x_{0})$, that is, 
\begin{equation*}
\hat{\xi}(x_{0})=\frac{1}{K}\sum_{k=1}^{K}\{\log (Y_{(k),[x_{0}]})-\log\left( Y_{(K+1),[x_{0}]}\right) \}.  \label{eq:xi}
\end{equation*}
Using similar lines of arguments to the proof of Theorem 2 in \citet{sasaki2022fixed}, we have that
\begin{equation}
\sqrt{K}\left( \hat{\xi}(x_{s})-\xi (x_{s})\right) \overset{d}{\rightarrow } \mathcal{N}\left( 0,\xi (x_{s})^{2}\right)
\label{eq:xi}
\end{equation}
for $s=1,0$.
In particular, we use $\alpha (x_{0})$ and $\gamma (x_{0})$ to denote their $\gamma (x_{0})$ and $\tilde{\gamma}\left( x_{0}\right) $, respectively.
Besides, their asymptotic bias term $\mu _{H}$ becomes zero by our Condition 1.3. 
Thus, the asymptotic distribution of $\hat{\alpha}\left( x_{s}\right) $ for $s=1,0$ is derived by the delta method.

Now, it remains to show the asymptotic independence between $\hat{\xi}(x_{1})$ and $\hat{\xi}(x_{0})$ for any $x_{1}\neq x_{0}$. 
Given their joint asymptotic Gaussianity, it suffices to show that 
\begin{equation*}
K\cdot \text{Cov}\left[ \hat{\xi}(x_{1})-\xi (x_{1}),\hat{\xi}(x_{0})-\xi(x_{0})\right] \rightarrow 0,
\end{equation*}
which is proved as follows.

First, given the i.i.d. condition across $i$, we have that
\begin{equation}
\text{Cov}\left[ \log (Y_{i_{1},[x_{0}]}),\log (Y_{i_{2},[x_{1}]})\right] =0\text{ for any } i_{1}\neq i_{2}.  \label{indep across i}
\end{equation}
Moreover, given the i.i.d. condition across $j$, the standard argument in the induced order statistics literature (e.g., Lemma 3.1 in \citet{Bhattachary1984}) yields that, for any $i$, conditional on $\mathbf{X}_{i}\equiv \{X_{i,1},...,X_{i,J}\}$, the induced order statistics $\{Y_{i,(1)},...,Y_{i,(J)}\}$ are independent. 
Denote $X_{i,[x]}$ as the NN of $\{X_{i,j}\}_{j=1}^{J}$ to $x$. 
Therefore as long as $X_{i,[x_{0}]}\neq X_{i,[x_{1}]}$, we have that   
\begin{eqnarray}
&&\text{Cov}\left[ \log (Y_{i,[x_{0}]}),\log (Y_{i,[x_{1}]})\right]   \notag \\
&=&\mathbb{E}\left[ \text{Cov}\left[ \log (Y_{i,[x_{0}]}),\log(Y_{i,[x_{1}]})|X_{i,1},...,X_{i,J}\right] |\mathbf{X}_{i}\right]   \notag \\
&=&\mathbb{E}\left[ \text{Cov}\left[\log(Y_{i,j_{0}}),\log(Y_{i,j_{1}})|X_{i,j_{0}}=X_{i,[x_{0}]},X_{i,j_{1}}=X_{i,[x_{1}]}\right] |\mathbf{X}_{i}\right]   \notag \\
&=&0
\label{indep within i}
\end{eqnarray}
by the law of iterated expectations. 
Since $X_{ij}$ contains at least one continuous component, Lemma 1 in \citet{sasaki2022fixed} implies that $\left\vert \left\vert X_{i,[x_{0}]}-x_{0}\right\vert \right\vert =o_{a.s.}(1)$ and $\left\vert \left\vert X_{i,[x_{1}]}-x_{1}\right\vert \right\vert=o_{a.s.}\left( 1\right) $ as $J\rightarrow \infty $. 
This yields that, for any $i$,  
\begin{equation}
\mathbb{P}\left[ X_{i,[x_{0}]}=X_{i,[x_{1}]}\right] =o(1)\text{ as } J\rightarrow \infty.   \label{tie}
\end{equation}

Now, recall that $\hat{\xi}(x_{1})=K^{-1}\sum_{i=1}^{K}\{\log (Y_{(i),[x_{1}]})-\log \left( Y_{(K+1),[x_{1}]}\right) \}$, which involves $\{Y_{i_{1},[x_{1}]},...,Y_{i_{K+1},[x_{1}]}\}$ for some indices $i_{1},...,i_{K+1}\in \{1,...,I\}$. 
Symmetrically, we have that $\hat{\xi}(x_{0})=K^{-1}\sum_{i=1}^{K}\{\log (Y_{(i),[x_{0}]})-\log \left(Y_{(K+1),[x_{0}]}\right) \}$, which involves $\{Y_{i_{1}^{\prime},[x_{1}]},...,Y_{i_{K+1}^{\prime },[x_{1}]}\}$ for some indices $i_{1}^{\prime },...,i_{K+1}^{\prime }\in \{1,...,I\}$. 
By (\ref{indep across i}), the only potentially non-zero components in $\text{Cov}[ \hat{\xi}(x_{1}),\hat{\xi}\left( x_{0}\right) ] $ will be $\text{Cov}\left[ \log(Y_{i,[x_{1}]}),\log (Y_{i,[x_{0}]})\right] $ for $i\in $ $\{i_{1},...,i_{K+1}\}\cap \{i_{1}^{\prime},...,i_{K+1}^{\prime }\}$. However, these terms are still zero by (\ref{indep within i}) if $X_{i,[x_{0}]}\neq X_{i,[x_{1}]}$, which happens with probability approaching one by (\ref{tie}). 
This completes a proof of the proposition.
\end{proof}

\subsection{Simulation Analysis}\label{sec:simulation}

In this section, we present simulation studies of finite sample performance of the proposed estimation and inference method.
In each iteration, we randomly generate an $I \times J$ array $\{Y_{ij},X_{ij}\}$ according to the following two designs.
Design 1:
$
Y_{ij} \sim \text{Pareto}(\alpha(X_{ij}))
$
where
$
\alpha(x) = 1+10x
$
and
$
X_{ij} \sim \text{Uniform}(0,1).
$
Design 2:
$
Y_{ij} \sim \text{Pareto}(\alpha(X_{ij}))
$
where
$
\alpha(x) = 10*(x^2-x+1)
$
and
$
X_{ij} \sim \text{Uniform}(0,1).
$
We estimate $\alpha(x_0)$ by $\hat\alpha(x_0)$ for each of $x_0 \in \{0.1,0.2,\cdots,0.8,0.9\}$.
Each set of Monte Carlo simulations consists of 1000 iterations of this process. 

Tables \ref{tab:simulation1} and \ref{tab:simulation2} report the results across various values of $I = J \in \{500, 1000\}$ and $K \in \{10,20\}$ under Design 1 and Design 2, respectively.
Displayed statistics are the bias (Bias), standard deviation (SD), root mean square error (RMSE), and 95\% coverage frequency (95\% Cover).
Observe that the RMSE shrinks as $K$ increases.
Furthermore, the 95\% coverage frequencies are close to the nominal probability.
These results demonstrate desired finite sample performance of the proposed method of estimation and inference.

\begin{table}
	\centering
		\begin{tabular}{rccccccccc}
		\hline\hline
		  \multicolumn{1}{l}{$I=J=500$} & \multicolumn{9}{c}{$x_0$}\\
			\cline{2-10}
			\multicolumn{1}{l}{$K = 10$} & 0.100 & 0.200 & 0.300 & 0.400 & 0.500 & 0.600 & 0.700 & 0.800 & 0.900 \\\hline
			Bias      & 0.221 & 0.304 & 0.497 & 0.594 & 0.614 & 0.712 & 0.979 & 0.920 & 1.069\\
			SD        & 0.759 & 1.125 & 1.476 & 2.043 & 2.178 & 2.812 & 3.066 & 3.608 & 4.002\\
			RMSE      & 0.791 & 1.166 & 1.558 & 2.128 & 2.263 & 2.901 & 3.219 & 3.724 & 4.142\\
			95\% Cover& 0.948 & 0.957 & 0.963 & 0.952 & 0.963 & 0.952 & 0.956 & 0.950 & 0.949\\
		\hline
		  \multicolumn{1}{l}{$I=J=500$} & \multicolumn{9}{c}{$x_0$}\\
			\cline{2-10}
			\multicolumn{1}{l}{$K = 20$} & 0.100 & 0.200 & 0.300 & 0.400 & 0.500 & 0.600 & 0.700 & 0.800 & 0.900 \\\hline
			Bias      & 0.121 & 0.158 & 0.207 & 0.273 & 0.262 & 0.370 & 0.519 & 0.451 & 0.519\\
			SD        & 0.512 & 0.722 & 0.951 & 1.287 & 1.448 & 1.773 & 2.126 & 2.141 & 2.568\\
			RMSE      & 0.526 & 0.740 & 0.973 & 1.316 & 1.472 & 1.811 & 2.189 & 2.188 & 2.620\\
			95\% Cover& 0.949 & 0.966 & 0.956 & 0.937 & 0.959 & 0.947 & 0.945 & 0.954 & 0.943\\
		\hline
		  \multicolumn{1}{l}{$I=J=1000$} & \multicolumn{9}{c}{$x_0$}\\
			\cline{2-10}
			\multicolumn{1}{l}{$K = 10$} & 0.100 & 0.200 & 0.300 & 0.400 & 0.500 & 0.600 & 0.700 & 0.800 & 0.900 \\\hline
			Bias      & 0.244 & 0.333 & 0.580 & 0.574 & 0.643 & 0.861 & 0.813 & 0.981 & 1.205\\
			SD        & 0.815 & 1.233 & 1.589 & 1.996 & 2.357 & 2.780 & 2.947 & 3.287 & 3.809\\
			RMSE      & 0.851 & 1.277 & 1.692 & 2.077 & 2.443 & 2.910 & 3.057 & 3.431 & 3.995\\
			95\% Cover& 0.957 & 0.951 & 0.963 & 0.953 & 0.956 & 0.947 & 0.954 & 0.962 & 0.965\\
		\hline
		  \multicolumn{1}{l}{$I=J=1000$} & \multicolumn{9}{c}{$x_0$}\\
			\cline{2-10}
			\multicolumn{1}{l}{$K = 20$} & 0.100 & 0.200 & 0.300 & 0.400 & 0.500 & 0.600 & 0.700 & 0.800 & 0.900 \\\hline
			Bias      & 0.122 & 0.148 & 0.272 & 0.254 & 0.315 & 0.432 & 0.418 & 0.459 & 0.471\\
			SD        & 0.508 & 0.767 & 1.007 & 1.219 & 1.453 & 1.760 & 2.020 & 2.183 & 2.537\\
			RMSE      & 0.523 & 0.781 & 1.043 & 1.245 & 1.487 & 1.812 & 2.063 & 2.231 & 2.580\\
			95\% Cover& 0.944 & 0.950 & 0.957 & 0.947 & 0.958 & 0.946 & 0.953 & 0.965 & 0.947\\
		\hline\hline
		\end{tabular}
	\caption{Simulation results under Design 1.}
	\label{tab:simulation1}
\end{table}

\begin{table}
	\centering
		\begin{tabular}{rccccccccc}
		\hline\hline
		  \multicolumn{1}{l}{$I=J=500$} & \multicolumn{9}{c}{$x_0$}\\
			\cline{2-10}
			\multicolumn{1}{l}{$K = 10$} & 0.100 & 0.200 & 0.300 & 0.400 & 0.500 & 0.600 & 0.700 & 0.800 & 0.900 \\\hline
			Bias      & 1.020 & 0.921 & 0.989 & 0.967 & 0.872 & 0.702 & 0.856 & 0.986 & 1.000\\
			SD        & 3.483 & 3.164 & 3.060 & 2.895 & 2.957 & 2.922 & 3.083 & 3.332 & 3.689\\
			RMSE      & 3.629 & 3.295 & 3.216 & 3.052 & 3.083 & 3.005 & 3.200 & 3.475 & 3.822\\
			95\% Cover& 0.957 & 0.960 & 0.963 & 0.955 & 0.948 & 0.950 & 0.951 & 0.965 & 0.953\\
		\hline
		  \multicolumn{1}{l}{$I=J=500$} & \multicolumn{9}{c}{$x_0$}\\
			\cline{2-10}
			\multicolumn{1}{l}{$K = 20$} & 0.100 & 0.200 & 0.300 & 0.400 & 0.500 & 0.600 & 0.700 & 0.800 & 0.900 \\\hline
			Bias      & 0.523 & 0.408 & 0.448 & 0.514 & 0.375 & 0.335 & 0.397 & 0.576 & 0.527\\
			SD        & 2.266 & 2.026 & 1.979 & 1.881 & 1.822 & 1.784 & 1.887 & 2.105 & 2.252\\
			RMSE      & 2.325 & 2.066 & 2.030 & 1.950 & 1.860 & 1.815 & 1.929 & 2.183 & 2.313\\
			95\% Cover& 0.958 & 0.951 & 0.936 & 0.948 & 0.956 & 0.947 & 0.953 & 0.953 & 0.953\\
		\hline
		  \multicolumn{1}{l}{$I=J=1000$} & \multicolumn{9}{c}{$x_0$}\\
			\cline{2-10}
			\multicolumn{1}{l}{$K = 10$} & 0.100 & 0.200 & 0.300 & 0.400 & 0.500 & 0.600 & 0.700 & 0.800 & 0.900 \\\hline
			Bias      & 0.843 & 0.933 & 0.944 & 0.765 & 0.725 & 0.824 & 0.790 & 0.750 & 1.131\\
			SD        & 3.409 & 3.349 & 3.015 & 2.795 & 2.890 & 3.063 & 2.891 & 3.344 & 3.586\\
			RMSE      & 3.512 & 3.476 & 3.160 & 2.897 & 2.979 & 3.172 & 2.997 & 3.428 & 3.760\\
			95\% Cover& 0.951 & 0.961 & 0.970 & 0.956 & 0.952 & 0.950 & 0.965 & 0.951 & 0.951\\
		\hline
		  \multicolumn{1}{l}{$I=J=1000$} & \multicolumn{9}{c}{$x_0$}\\
			\cline{2-10}
			\multicolumn{1}{l}{$K = 20$} & 0.100 & 0.200 & 0.300 & 0.400 & 0.500 & 0.600 & 0.700 & 0.800 & 0.900 \\\hline
			Bias      & 0.474 & 0.401 & 0.446 & 0.359 & 0.339 & 0.413 & 0.440 & 0.324 & 0.495\\
			SD        & 2.256 & 2.130 & 2.069 & 1.831 & 1.910 & 1.962 & 1.913 & 2.033 & 2.299\\
			RMSE      & 2.306 & 2.168 & 2.116 & 1.866 & 1.940 & 2.005 & 1.963 & 2.059 & 2.352\\
			95\% Cover& 0.953 & 0.948 & 0.944 & 0.947 & 0.947 & 0.948 & 0.960 & 0.947 & 0.941\\
		\hline\hline
		\end{tabular}
	\caption{Simulation results under Design 2.}
	\label{tab:simulation2}
\end{table}

For comparison, we also implement the fully nonparametric method proposed by \citet{Gardes2008}. 
In particular, given a random sample of $(Y_i, X_i)$ of size $N$ and a pre-determined bandwidth $b$, we follow \citet{Gardes2008} to select all the observations that satisfy $|X_i - x_0| \le b$, where $x_0$ again denotes the conditional value. 
Given this selected subsample, we sort the $Y_i$'s and construct the standard Hill estimator based on the largest $K+1$ order statistics as in \eqref{eq:hill}. 
Since there is no theoretical justification of the choice of $b$, we implement various values of $b \in \{0.01, 0.025, 0.05, 0.075, 0.1, 0.125, 0.15, 0.175, 0.2\}$ for sensitivity check. 
We generate data from Design 1 described above and set $x_0 = 0.5$. 
The sample size is $N = I*J \in \{500^2, 1000^2\}$. 
Table \ref{tab:simulation3} presents the results with 1000 iterations. 
We summarize the findings as follows. 
First, the choice of bandwidth has a substantial effect on the performance of this nonparametric estimator. 
The bias is small only when the bandwidth is within a narrow window, which theoretically depends on the unknown higher-order parameters. 
Second, compared with Table \ref{tab:simulation1}, this nonparametric estimator has larger bias and standard deviations than our proposed method when the bandwidth is small enough to guarantee the correct coverage. 

\begin{table}
	\centering
		\begin{tabular}{rccccccccc}
		\hline\hline
		  \multicolumn{1}{l}{$N=500^2$} & \multicolumn{9}{c}{$b$}\\
			\cline{2-10}
			\multicolumn{1}{l}{$K = 10$} & 0.010 & 0.025 & 0.050 & 0.075 & 0.100 & 0.125 & 0.150 & 0.175 & 0.200  \\\hline
			Bias           & 0.627 & 0.562 & 0.479 & 0.428 & 0.296 & -0.064 & -0.360 & -0.549 & -0.994\\
			SD             & 2.349 & 2.238 & 2.213 & 2.186 & 2.153 & 2.068 & 2.058 & 1.952 & 1.849\\
			RMSE         & 2.430 & 2.307 & 2.263 & 2.227 & 2.172 & 2.068 & 2.088 & 2.026 & 2.098\\
			95\% Cover& 0.956 & 0.947 & 0.946 & 0.947 & 0.933 & 0.902 & 0.867 & 0.856 & 0.763\\
		\hline
		  \multicolumn{1}{l}{$N=500^2$} & \multicolumn{9}{c}{$b$}\\
			\cline{2-10}
			\multicolumn{1}{l}{$K = 20$} & 0.010 & 0.025 & 0.050 & 0.075 & 0.100 & 0.125 & 0.150 & 0.175 & 0.200  \\\hline
			Bias      & 0.284 & 0.334 & 0.118 & 0.086 & -0.106 & -0.383 & -0.654 & -0.916 & -1.269 \\
			SD        & 1.482 & 1.549 & 1.453 & 1.501 &	1.444 & 1.359 &	1.196 & 1.184 &	1.086\\
			RMSE      & 1.509 & 1.584 & 1.457 & 1.503 & 1.448 & 1.411 & 1.362 & 1.496 & 1.670\\
			95\% Cover& 0.952 & 0.950 & 0.937 & 0.937 & 0.927 & 0.881 & 0.852 & 0.770 & 0.656\\
		\hline
		  \multicolumn{1}{l}{$N=1000^2$} & \multicolumn{9}{c}{$b$}\\
			\cline{2-10}
			\multicolumn{1}{l}{$K = 10$} & 0.010 & 0.025 & 0.050 & 0.075 & 0.100 & 0.125 & 0.150 & 0.175 & 0.200  \\\hline
			Bias      & 0.842 & 0.501 & 0.457 & 0.294 & 0.088 & -0.238 & -0.565 & -0.804 & -1.024 \\
			SD        & 2.488	& 2.283 & 2.249 & 2.200 & 2.128 & 2.112 & 2.037 & 1.842 & 1.780 \\
                        RMSE   & 2.625 & 2.336 & 2.294 & 2.219 & 2.129 & 2.124 & 2.113 & 2.009 & 2.053 \\
                        95\% Cover & 0.966 & 0.951 & 0.950 & 0.933 & 0.921 & 0.891 & 0.828 & 0.806 & 0.758 \\
		\hline
		  \multicolumn{1}{l}{$N=1000^2$} & \multicolumn{9}{c}{$b$}\\
			\cline{2-10}
			\multicolumn{1}{l}{$K = 20$} & 0.010 & 0.025 & 0.050 & 0.075 & 0.100 & 0.125 & 0.150 & 0.175 & 0.200  \\\hline
			Bias      & 0.360 & 0.350 & 0.143 & 0.022 & -0.219 & -0.449 & -0.713 & -1.015 & -1.360 \\
                        SD        & 1.490 & 1.507 & 1.410 & 1.428 & 1.447 & 1.346 & 1.321 & 1.170 & 1.089 \\
                        RMSE    & 1.532 & 1.546 & 1.416 & 1.428 & 1.463 & 1.418 & 1.501 & 1.548 & 1.742 \\
                        95\% Cover& 0.961 & 0.958 & 0.947 & 0.930 & 0.892 & 0.869 & 0.807 & 0.747 & 0.631 \\

		\hline\hline
		\end{tabular}
	\caption{Simulation results of the nonparametric estimator under Design 1.}
	\label{tab:simulation3}
\end{table}

\subsection{Choice of $I$, $J$, and $K$}\label{sec:choice}

In this appendix section, we present a rule of choosing $I$, $J$, and $K$ in practice when a random sample of size $N$ is available.
For simplicity, we consider the Pareto distribution family as a reference.
In this case, we have $\gamma(x_{0}) = \infty$ in Condition \ref{cond:1}, and hence $\lambda \ge 1$.
Since $K$ is the effective sample size in the asymptotic normality result and $K$ in turn increasing in $I$, the goal is to choose $\lambda \ge 1$ such that $I$ is large.
By Condition \ref{cond:1}.4, such a choice is $\lambda = 1$.
We therefore suggest to set $I = J = \sqrt{N}$.

Once $I$ has been chosen, then we can choose $K$ by adapting the diagnostic method proposed by \citet{GuillouHall2001}. 
We consider the case of estimating $\alpha (x_{0})$ here. 
Define $Z_{i}=i\log(Y_{(i),[x_{0}]}/Y_{(i+1),[x_{0}]})$ for $i=1,\ldots ,I$. 
Suppose that $Y_{i,[x_{0}]}$ is exactly Pareto distributed with exponent $\alpha(x_{0})=:1/\xi (x_{0})$.
Then, $Z_{i}$ should be i.i.d. with the exponential distribution and satisfies $\mathbb{E}\left[ Z_{i}\right] =\xi (x_{0})$ and $\mathbb{V}\left[ Z_{m+j}\right] =\xi (x_{0})^{2}$. 
For any given $K$ and any antisymmetric weights $\{w_{i}\}_{i=1}^{K}$\ such that $w_{i}=-w_{K-i+1}$\ and $\sum_{i=1}^{K}w_{i}=0$, the weighted average statistic $U_{K}\equiv\sum_{i=1}^{K}w_{i}Z_{i}$ should have zero mean and variance $\sum_{i=1}^{K}w_{i}^{2}\xi (x_{0})^{2}$. 
Therefore, we can construct   
\begin{equation*}
\mathcal{T}_{K}\equiv \left( \sum_{i=1}^{K}w_{i}^{2}\right) ^{-1/2}\hat{\xi}%
(x_{0})^{-1}U_{K},
\end{equation*}%
which has zero mean and unit variance if $Y_{i,[x_{0}]}$ is exactly Pareto and if $\hat{\xi}(x_{0})=\xi (x_{0})$.

Now under Condition 1, $\{Z_{i}\}_{i=1}^{K}$\ should be approximately independent and exponentially distributed. 
Then the above properties of $\mathcal{T}_{K}$ hold asymptotically, following a similar argument to that in \citet{GuillouHall2001}. 
When such an approximation performs well for a certain $K$, we expect the fluctuation of $\mathcal{T}_{K}^{2}$\ to be small.
Accordingly, we define the following criteria based on a moving average of $\mathcal{T}_{K}^{2}$: 
\begin{equation*}
\mathcal{C}_{K}=\left( \left( 2l+1\right) ^{-1}\sum_{j=-l}^{l}\mathcal{T}%
_{K+j}^{2}\right) ^{1/2},
\end{equation*}%
where $l$ equals the integer part of $K/2$. 
When $K$ is too large relative to $I$, the Pareto approximation incurs a larger bias, and hence $\mathcal{C}_{K}$ exceeds one by a larger magnitude. 
To obtain an implementable rule, we follow \citet{GuillouHall2001} to use $w_{i}=\text{sgn}\left( K-2i+1\right)\left\vert k-2i+1\right\vert $ and propose to choose the smallest $K$ that
satisfies $\mathcal{C}_{t}>1$ for all $t\geq K$, that is,%
\begin{equation}
K^{\ast }=\min_{1\leq K\leq I}\{K:\mathcal{C}_{t}>1\text{ for all }t\geq K\}.  
\label{k choice}
\end{equation}

To incorporate the uncertainty induced by the random splitting of data, we follow a procedure suggested in Section 3.4 of \citet{CCDDHNR}.
Suppose that we obtain $S$ estimates $\{\hat\alpha_s(x_0)\}_{s=1}^S$ by \eqref{eq:hill} with $S$ times of random splitting.
For point estimation, we use 
$
\bar\alpha_S(x_0) = \text{median}\{ \hat\alpha_s(x_0) \}_{s=1}^S.
$
For variance estimators, we use
\begin{equation*}
\hat\sigma_S^2(x_0) = \text{median}\left\{ \hat\alpha_s^2(x_0) + (\hat\alpha_s(x_0)-\bar\alpha_S(x_0))^2 \right\}_{s=1}^S
\end{equation*}
by accounting for the variation introduced by random splitting.
We use $S=1000$ in our empirical analyses.


\subsection{Refinement Based on the Fixed-$K$ Asymptotics}\label{sec:adjustments}
We close the section by introducing a refinement of our proposed method when the sample size is only moderate. 
Condition 1.3 requires that $K\rightarrow\infty$ and $K/n\rightarrow 0$, which could imply a delicate balance in choosing $K$. 
Therefore, when the sample size is only moderate, it is well known in the literature that the choice of $K$ could be challenging \citep[e.g.,][]{MuellerWang2017}. 
As an alternative, we may rely on the recently developed fixed-$k$ asymptotic framework, which assumes $K$ remains fixed as $n\rightarrow\infty$. 
Note that this is mainly for completeness since the sample size in our data set is reasonably large. 

Recall that
\begin{equation*}
\mathbf{Y}(x_{0})=\{Y_{(1),[x_{0}]},Y_{(2),[x_{0}]},...,Y_{(K+1),[x_{0}]}\}.
\end{equation*}
denote the largest $K+1$ order statistics among the $I$ nearest neighbors to $x_0$. 
Condition 1.1 imples that there exist sequences of constants $a_{I}$ and $b_{I}$ such that for every $v$, 
\begin{equation*}
\lim_{I \rightarrow \infty}F_{Y|X=x_{0}}(a_{I}v+b_{I})^I = G_{\xi }(v),
\end{equation*}
where 
\begin{equation}
G_{\xi }(v)=\left\{ 
\begin{array}{ll}
\exp (-(1+\xi v)^{-1/\xi })\text{, } & 1+\xi v>0\text{, for }\xi \neq 0 \\ 
\exp (-e^{-v})\text{, } & v\in \mathbb{R} \text{, }\xi =0
\end{array}
\right.  \label{def_G}
\end{equation}
is the so-called generalized extreme value distribution. 
In our heavy tail setup, we have that $\xi=\xi(x_0)=1/\alpha(x_0)>0$.

Given Conditions 1.1 and 1.2, the proof of Theorem 1 in \citet{sasaki2022fixed} shows that there exist sequences of constants $a_I$ and $b_I$ such that for any fixed $K$, 
\begin{equation}
\frac{\mathbf{Y}(x_0)-b_{I}}{a_{I}}\overset{d}{\rightarrow }\mathbf{V}=\left( 
\begin{array}{c}
V_{1} \\ 
\vdots \\ 
V_{K+1}%
\end{array}%
\right) \text{, as }I,J\rightarrow \infty ,  \label{evt_k}
\end{equation}%
where the joint probability density function (PDF)~of $\mathbf{V}$ is given by 
\begin{equation}
f_{\mathbf{V}}(v_{1},\ldots ,v_{K+1};\xi )=G_{\xi }(v_{K+1})\prod_{i=1}^{K+1}g_{\xi }(v_{i})/G_{\xi }(v_{i})  \label{evt_pdf}
\end{equation}%
for $v_{K+1}\leq v_{K}\leq \ldots \leq v_{1}$ with $g_{\xi }(v)=\partial
G_{\xi }(v)/\partial v$, and zero otherwise.

Note that the constants $a_{I}$ and $b_{I}$ depend on $\xi $ (and equivalently $\alpha(x_0)$), and a precise estimation of them is challenging, especially when the sample size is only moderate. 
To sidestep this problem, we can consider the following self-normalized statistics 
\begin{equation*}
\mathbf{Y}^{\ast}(x_0) = \frac{\mathbf{Y}(x_0) - Y_{(K+1),[x_{0}]}}{Y_{(1),[x_{0}]} - Y_{(K+1),[x_{0}]}},
\end{equation*}
which converges in distribution to 
\begin{equation}\label{eq:vstar}
\mathbf{V}^{\ast } = \left( 1,\frac{V_{2}-V_{K+1}}{V_{1}-V_{K+1}},...,0\right) . 
\end{equation}
Note that the first and the last components are respectively one and zero by construction. 
The limiting density $f_{\mathbf{V}^{\ast}}$ can be easily derived from \eqref{evt_pdf} \citep[e.g.,][equation (3.2)]{MaSasakiWang2022}.
Since this limiting density is fully characterized by $\xi=1/\alpha(x_0)$, we can perform the (generalized) likelihood ratio test for the problem \eqref{hypo}. 
In particular, we can test, for $r=2,3,4$,
\begin{equation*}
H_0: \xi = 1/r \text{ against } H_1: \xi > 1/r 
\end{equation*}
using the test statistic similar to \citet[][equation (2.3)]{MaSasakiWang2022}.
Such a fixed-$K$ test controls size asymptotically as $I, J \rightarrow \infty$, which performs well when the sample size is only moderate.


\section{Additional Results with NESPD}\label{sec:additional_results}
In this section, we present supplementary results using the UK NESDP data set to illustrate the robustness of our main findings. 

\subsection{Density Estimates}\label{app: density}
We provide more information for the kernel estimates that we plot in the densities figures, Figures \ref{fig:nespd_density_2007} and \ref{fig:nespd_density_2015}, which are obtained with the Gaussian smoothing kernel. 
Figure \ref{fig:nespd_density_all_weightings_2007} illustrates the kernel density estimates under different weighting functions, for men in the period 2007--2008. We use the following weighting functions: Gaussian (default in the paper), Epanechnikov, triangular, rectangular, and cosine, as well as the histogram of the data. The kernel estimates do not vary significantly under the different weighting functions. 

\begin{figure}
	\centering
	\includegraphics[width=0.3\textwidth]{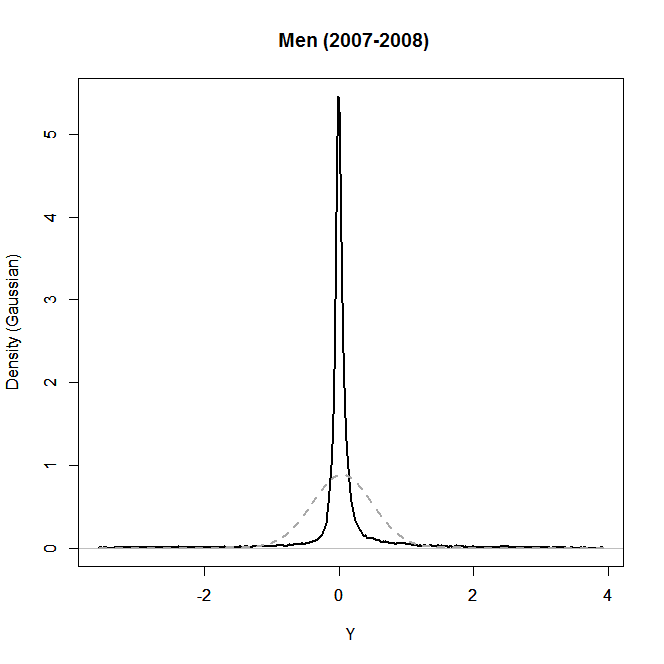}
	\includegraphics[width=0.3\textwidth]{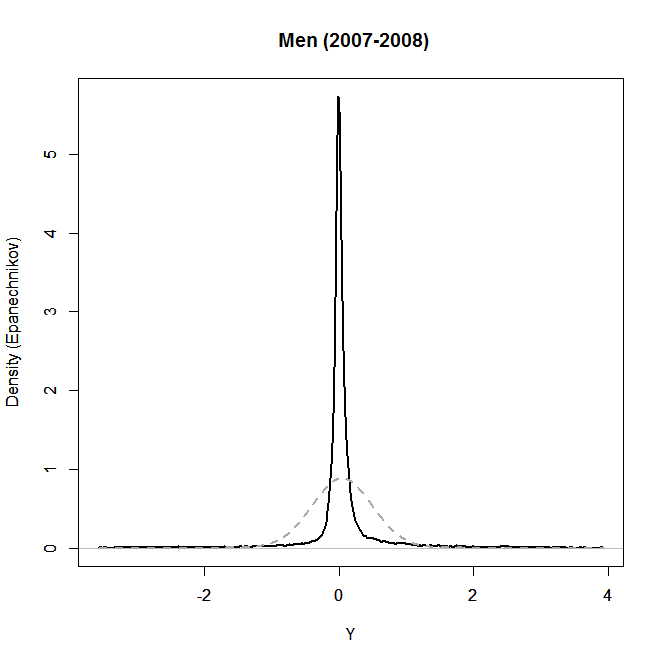}
	\includegraphics[width=0.3\textwidth]{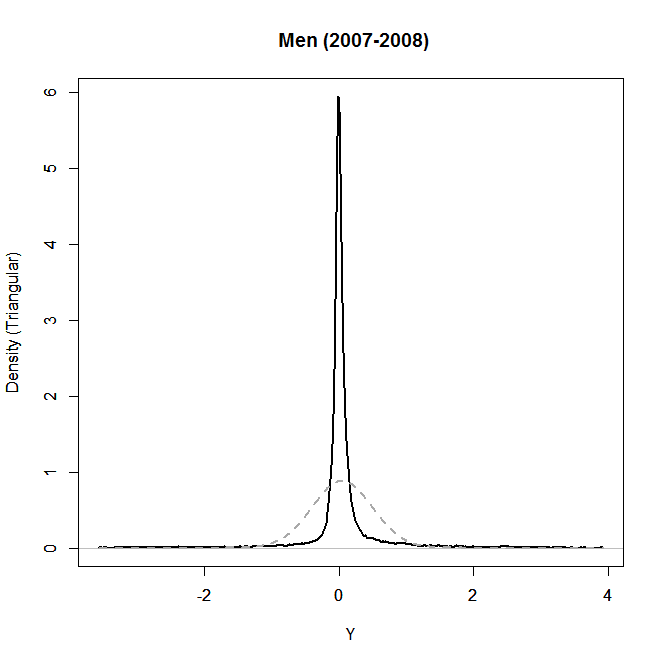}
	\includegraphics[width=0.3\textwidth]{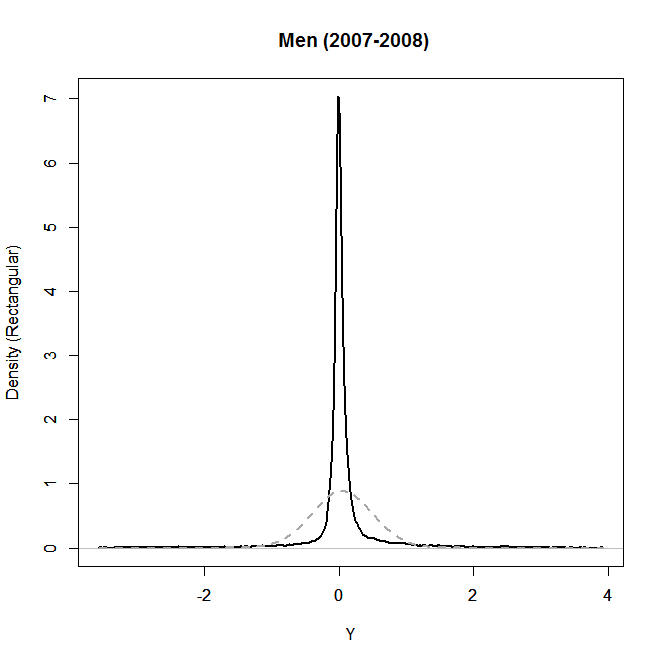}
	\includegraphics[width=0.3\textwidth]{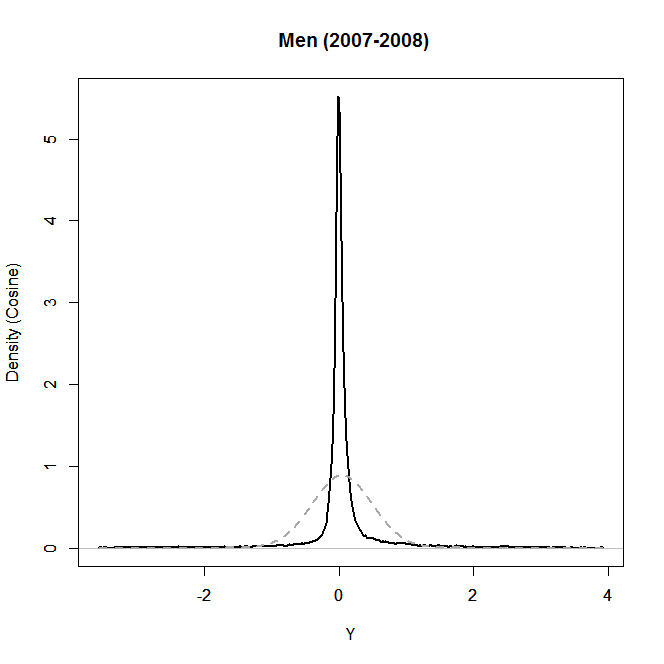}
	\includegraphics[width=0.3\textwidth]{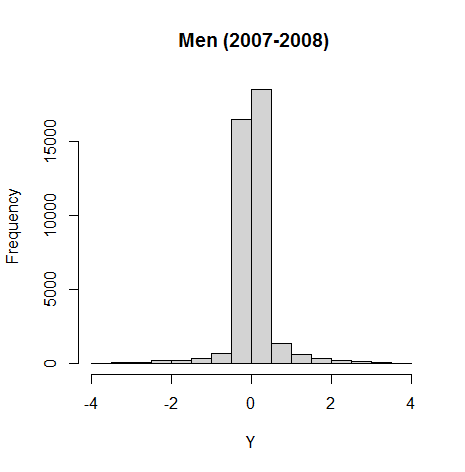}
	\caption{Kernel density estimates (black line) of $Y$, which is defined as a one-year change in log earnings, in 2007 in the NESPD. The panels show the density of men with five different weighting schemes as well as the histogram of the data. Also shown in gray dashed lines are the normal density fit to data. Number of individuals: 38,955 men.}
	\label{fig:nespd_density_all_weightings_2007}
\end{figure}

\subsection{Tail Earnings Risk Over Time}\label{sec:additional_years}
In Section \ref{sec:results_admin} in the main text, we focus on the period 2007--2008 of the great recession and the period 2015--2016 of growth for conservation of space.
In this appendix section, we repeat the econometric analysis for every other period, in addition to the benchmark years: 2005--2006, 2009--2010, 2011--2012, and 2013--2014 to demonstrate the robustness of the observed patterns reported in the main text.\footnote{Due to space constraints, we provide results for alternate periods. The results for other years 
can be obtained upon request from the authors.}
Figures \ref{fig:nespd_2005}, \ref{fig:nespd_2009}, \ref{fig:nespd_2011}, and \ref{fig:nespd_2013} illustrate the estimates of the conditional Pareto exponents $\alpha(x_0)$ (in black lines) along with the upper bounds of their one-sided 95\% confidence intervals (in gray lines). 
In each figure, the left (respectively, right) column shows the results for men (respectively, women). 
The top, middle, and bottom panels show results for 30-, 40- and 50-year-old individuals. 
The findings are similar to those in Section \ref{sec:empirical}. 

\begin{figure}
	\centering
		\includegraphics[width=0.49\textwidth]{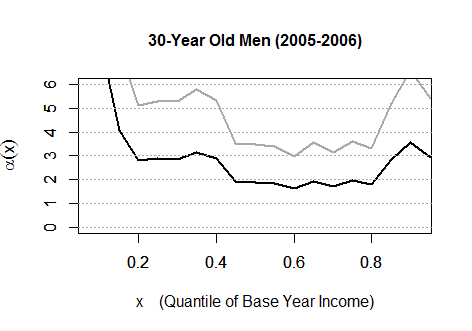}
		\includegraphics[width=0.49\textwidth]{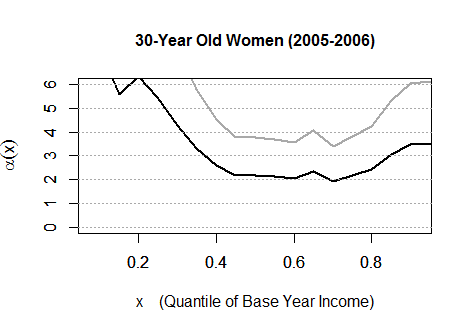}
		\includegraphics[width=0.49\textwidth]{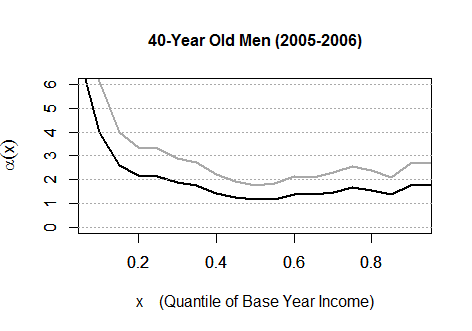}
		\includegraphics[width=0.49\textwidth]{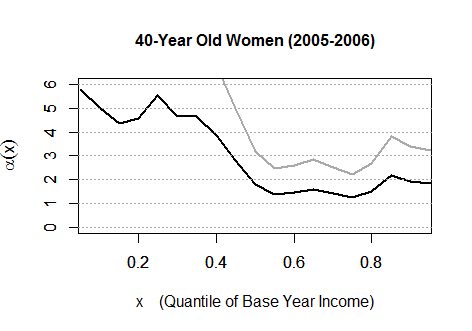}
		\includegraphics[width=0.49\textwidth]{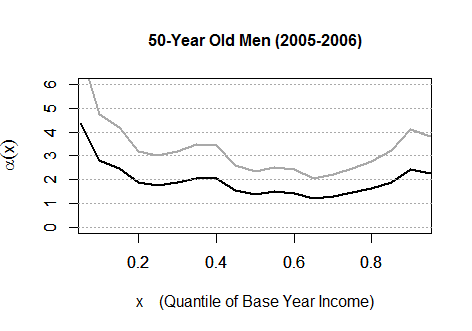}
		\includegraphics[width=0.49\textwidth]{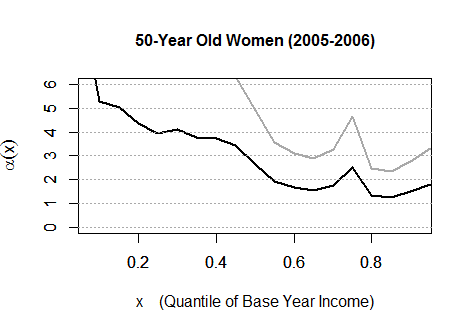}
	\caption{Estimates (black lines) and the one-sided 95\% confidence intervals (gray lines) of the Pareto exponents $\alpha(x_0)$ of the conditional tail risk for men (left) and women (right) based on the NESPD in the period 2005--2006. The left (respectively, right) column shows the results for men (respectively, women). The top, middle, and bottom panels show results for 30-, 40- and 50-year-old individuals. Number of individuals: 48,542 men (47,851 women).}
	\label{fig:nespd_2005}
\end{figure}
\begin{figure}
	\centering
		\includegraphics[width=0.49\textwidth]{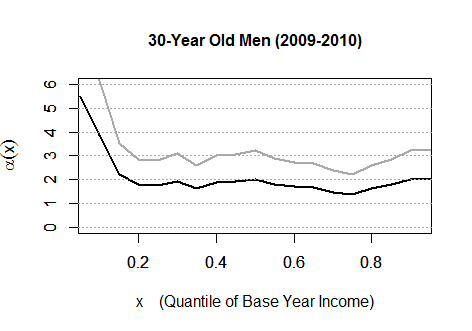}
		\includegraphics[width=0.49\textwidth]{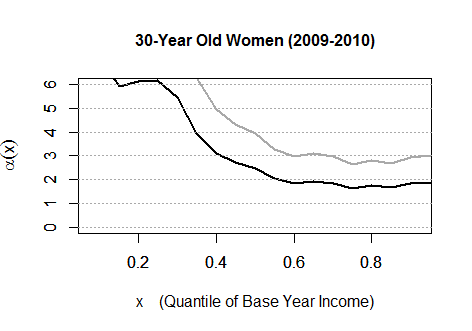}
		\includegraphics[width=0.49\textwidth]{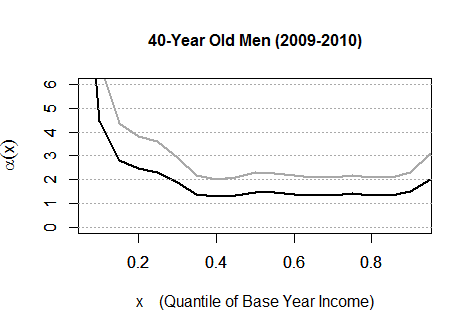}
		\includegraphics[width=0.49\textwidth]{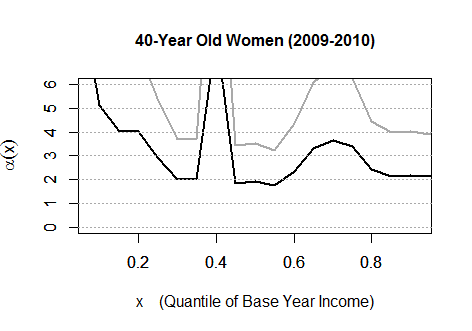}
		\includegraphics[width=0.49\textwidth]{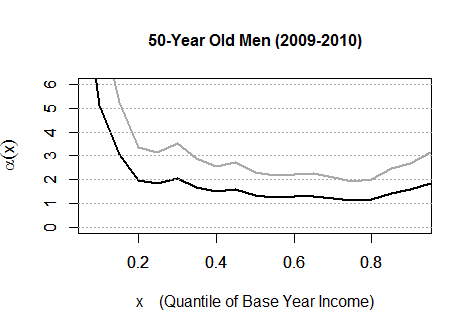}
		\includegraphics[width=0.49\textwidth]{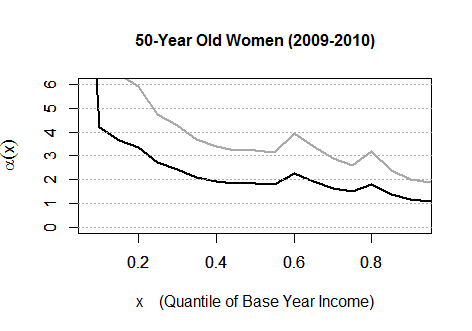}
	\caption{Estimates (black lines) and the one-sided 95\% confidence intervals (gray lines) of the Pareto exponents $\alpha(x_0)$ of the conditional tail risk for men (left) and women (right) based on the NESPD in the period 2009--2010. The left (respectively, right) column shows the results for men (respectively, women). The top, middle, and bottom panels show results for 30-, 40- and 50-year-old individuals. Number of individuals: 48,368 men (50,372 women).}
	\label{fig:nespd_2009}
\end{figure}
%

\begin{figure}
	\centering
		\includegraphics[width=0.49\textwidth]{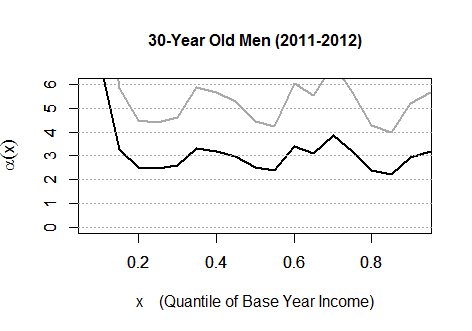}
		\includegraphics[width=0.49\textwidth]{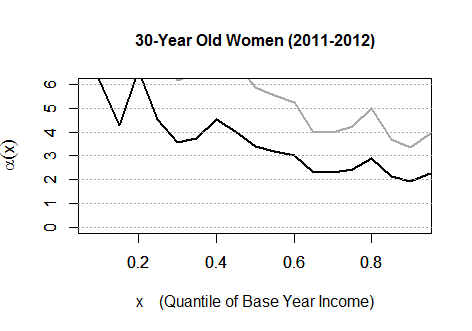}
		\includegraphics[width=0.49\textwidth]{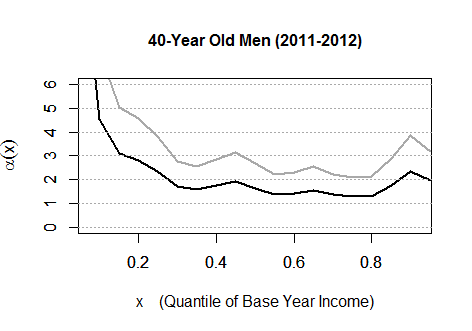}
		\includegraphics[width=0.49\textwidth]{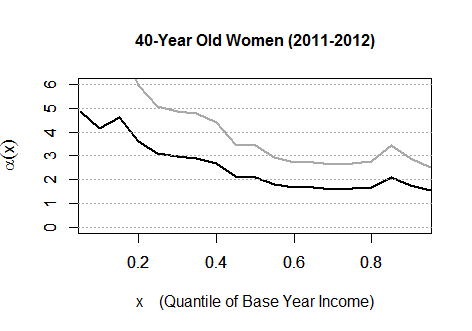}
		\includegraphics[width=0.49\textwidth]{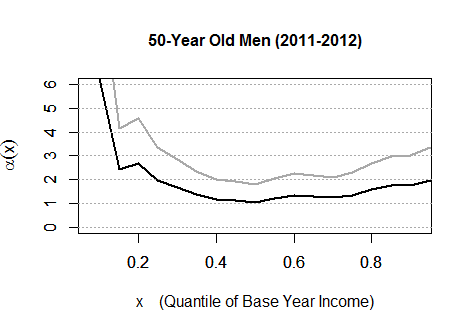}
		\includegraphics[width=0.49\textwidth]{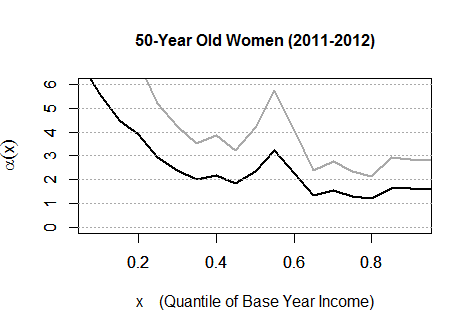}
	\caption{Estimates (black lines) and the one-sided 95\% confidence intervals (gray lines) of the Pareto exponents $\alpha(x_0)$ of the conditional tail risk for men (left) and women (right) based on the NESPD in the period 2011--2012. The left (respectively, right) column shows the results for men (respectively, women). The top, middle, and bottom panels show results for 30-, 40- and 50-year-old individuals. Number of individuals: 49,720 men (52,874 women).}
	\label{fig:nespd_2011}
\end{figure}


\begin{figure}
	\centering
		\includegraphics[width=0.49\textwidth]{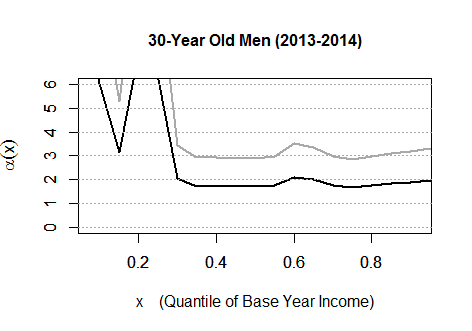}
		\includegraphics[width=0.49\textwidth]{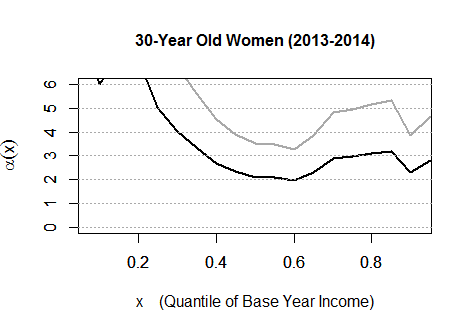}
		\includegraphics[width=0.49\textwidth]{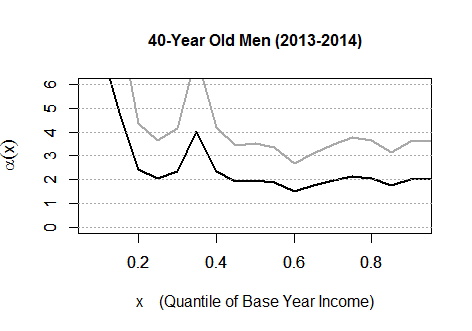}
		\includegraphics[width=0.49\textwidth]{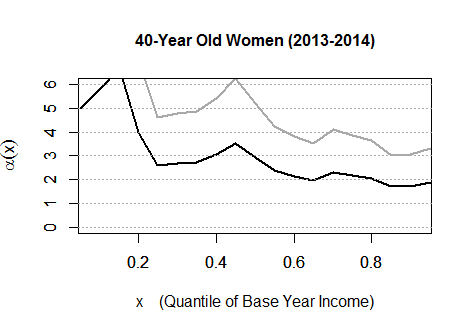}
		\includegraphics[width=0.49\textwidth]{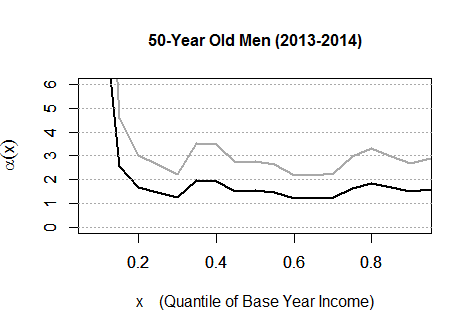}
		\includegraphics[width=0.49\textwidth]{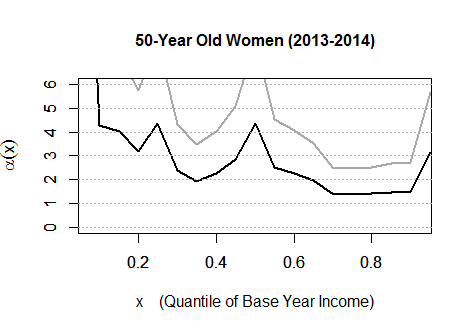}
	\caption{Estimates (black lines) and the one-sided 95\% confidence intervals (gray lines) of the Pareto exponents $\alpha(x_0)$ of the conditional tail risk for men (left) and women (right) based on the NESPD in the period 2013--2014. The left (respectively, right) column shows the results for men (respectively, women). The top, middle, and bottom panels show results for 30-, 40- and 50-year-old individuals. Number of individuals: 49,455 men (52,586 women).}
	\label{fig:nespd_2013}
\end{figure}


\subsection{Conditioning Variable: Average of Lagged Earnings}\label{app:avg}
The main text focuses on single-year earnings as the conditioning variable.
The current appendix section conducts robustness checks with an alternative definition of the conditioning variable.
In Figure \ref{fig:nespd_2007_avg3}, we repeat the analysis using as the conditioning variable the average of the three previous earnings from $t-1$, going backward, instead of base year earnings. Thus, we consider the average of earnings over the period 2004--2006 as the conditioning variable and restrict the sample to be the balanced panel of men with non-missing observations from $t-3$ to $t+1$ with $t=2007$ in this case. Some quantitative differences arise with respect to the distribution of lagged earnings compared to the main results shown in Figures \ref{fig:nespd_2007}. However, note that the sample size almost halves when restricting the panel to be balanced in the five years from 2004 to 2008, which we use to produce Figure \ref{fig:nespd_2007_avg3} (20,682 men and 19,864 women for 2007--2008). The differences with respect to the benchmark sample are, at least partly, due to the reduction in sample size.\footnote{Indeed, we obtain similar findings to those obtained conditioning on the average of the three previous earnings, when repeating the analysis with the restricted sample size of Figure \ref{fig:nespd_2007_avg3}, but keeping the conditioning variable to be lagged earnings as in Figure \ref{fig:nespd_2007}.}
\begin{figure}
	\centering
	\includegraphics[width=0.49\textwidth]{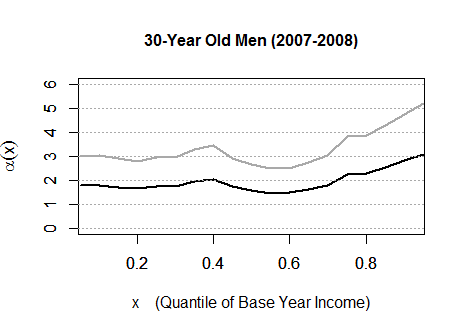}
	\includegraphics[width=0.49\textwidth]{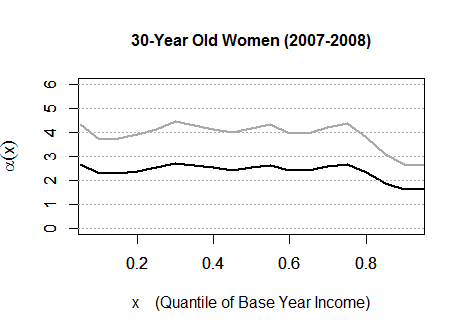}
	\includegraphics[width=0.49\textwidth]{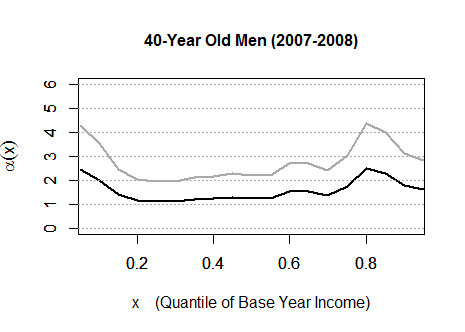}
	\includegraphics[width=0.49\textwidth]{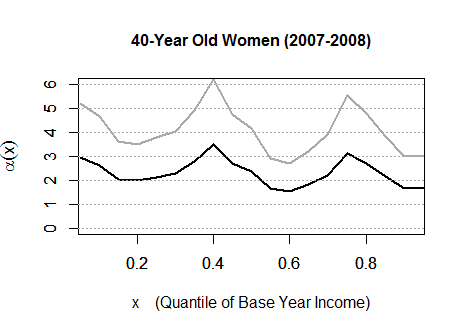}
	\includegraphics[width=0.49\textwidth]{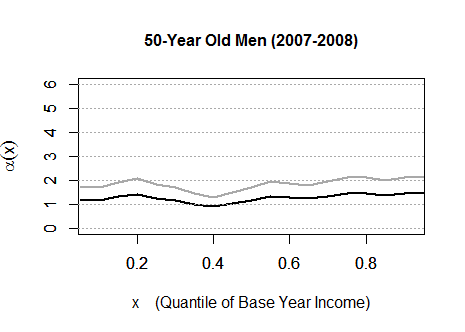}
	\includegraphics[width=0.49\textwidth]{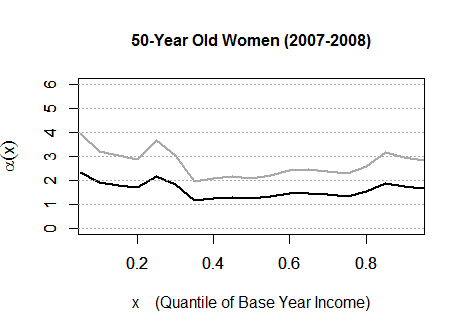}
	\caption{Estimates (black lines) and the one-sided 95\% confidence intervals (gray lines) of the Pareto exponents $\alpha(x_0)$ of the conditional tail risk for men (left) and women (right) based on the NESPD in the period 2007--2008. The left (respectively, right) column shows the results for men (respectively, women). The top, middle, and bottom panels show the results for 30-, 40- and 50-year-old individuals. The conditioning variable is the average of the 3 previous earnings, from $t-1$ going backwards. Number of individuals: 20,682 men and 19,864 women.}
	\label{fig:nespd_2007_avg3}
\end{figure}

\section{Panel Study of Income Dynamics}\label{sec:psid}

\subsection{Data}\label{sec:psid_data}

We also apply our methodology to another data set, the US Panel Study of Income Dynamics (PSID), which is one of the most commonly used data sets for empirical research on income dynamics.
The portion of data for the pair, 2007 and 2009, of years across the period of the great recession is extracted for our use.
We select the subsample of both men and women whose total taxable incomes were recorded and non-zero in both 2007 and 2009 and were aged between 25 and 64 in 2007.
This sample selection leaves 271 individuals.
For the purpose of comparisons, we also use the portion of data for the pair, 2017 and 2019, of the most recent survey years at the time of our writing of this paper.
Note that this period is associated with positive economic growth in the United States.
The sample selection procedure described above leaves 389 individuals for this period.

We define a measure $Y$ as the absolute difference of the log income in 2007 and the log income in 2009 (i.e., two-year income growth rate).
Similarly, we also construct this variable for the period between 2017 and 2019.
Figures \ref{fig:psid_density_2007} and \ref{fig:psid_density_2017} display kernel density estimates of the income measure $Y$ for the period 2007--2009 and the period 2017--2019, respectively.
In each figure, the left and right panels show the densities for men and women, respectively.
Also shown in gray dashed lines are the normal density plots fit to data.
Observe that each kernel density exhibits a large spike in the middle of the distribution sticking upward out of the reference normal density.
Moreover, each kernel density has heavier tails compared to the reference normal density.
These features of the estimated densities evidence that the actual distributions of $Y$ indeed have heavier tails than normal distributions, as documented in the previous literature.
That said, as emphasized in the introductory section, nonparametric density plots cannot informatively demonstrate evidence of heavy tails.

\begin{figure}
	\centering
		\includegraphics[width=0.49\textwidth]{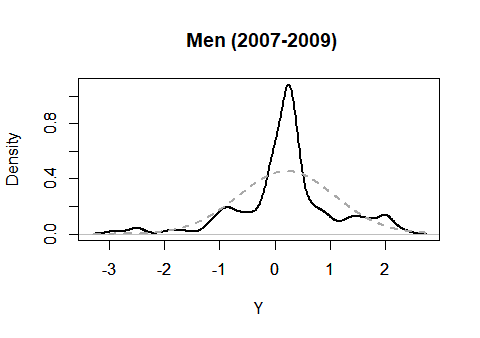}
		\includegraphics[width=0.49\textwidth]{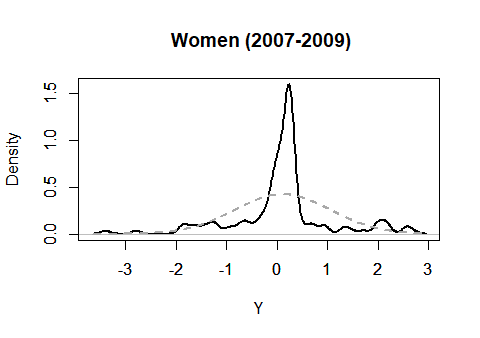}
	\caption{Kernel density estimates (black line) of the risk measure $Y$ in 2007 in the PSID. The right (respectively, left) panel show the density of men (respectively, women). Also shown in gray dashed lines are the normal density fit to data.}
	\label{fig:psid_density_2007}
\end{figure}

\begin{figure}
	\centering
		\includegraphics[width=0.49\textwidth]{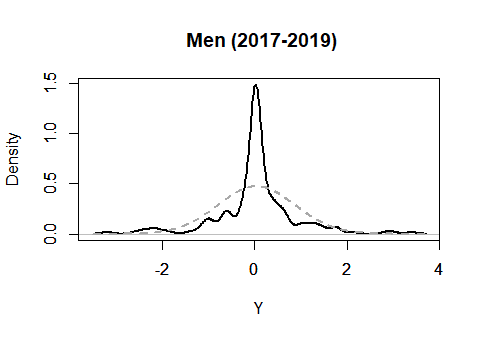}
		\includegraphics[width=0.49\textwidth]{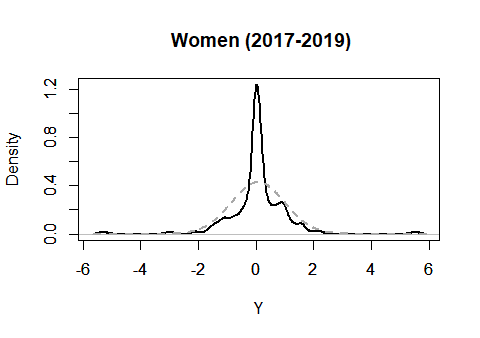}
	\caption{Kernel density estimates (black line) of the risk measure $Y$ in 2017 in the PSID. The right (respectively, left) panel show the density of men (respectively, women). Also shown in gray dashed lines are the normal density fit to data.}
	\label{fig:psid_density_2017}
\end{figure}

We analyze the heaviness of the tails of the conditional distributions of $Y$ given the income level and age in the base year (i.e., the base year is 2007 for the period 2007--2009 and it is 2017 for the period 2017--2019).

\subsection{Results}\label{sec:results_psid}

Applying the method introduced in Section \ref{sec:econometric_method} (and also Appendix \ref{sec:additional_details}) to the US Panel Study of Income Dynamics (PSID) described in Appendix \ref{sec:psid_data}, we analyze the conditional tail risk of income of adult individuals in the United States.
We define $Y$ by the absolute difference of the log income in 2007 and the log income in 2009 for our baseline analysis.
For the conditioning variables $X$, we include the quantile of income level and the age of the individual in the base year (2007), following \citet{guvenen2021data}.
With this setting, we study the conditional Pareto exponent $\alpha(x_0)$ for each point $x_0$ of income levels from $\{0.05,0.10,\cdots,0.90,0.95\}$ (in quantile) and ages from $\{30,40,50\}$ for each of men and women.

Figure \ref{fig:psid_2007} illustrates the estimates of the conditional Pareto exponents $\alpha(x_0)$ (in black lines) along with the upper bounds of their one-sided 95\% confidence intervals (in gray lines) for 30-, 40- and 50-year-old individuals, respectively.
The top (respectively, bottom) panel shows results for men (respectively, women) in each figure.
We observe different patterns of heterogeneous income risk across age, income, and gender groups and describe them in the order of age and gender. 

\begin{figure}
	\centering
		\includegraphics[width=0.49\textwidth]{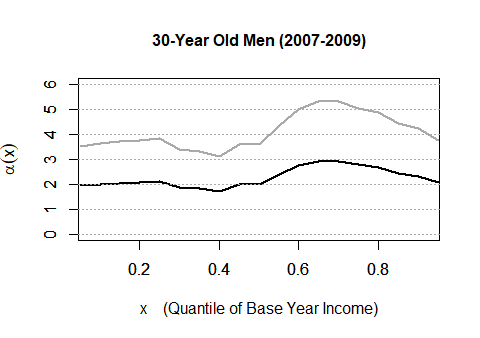}
		\includegraphics[width=0.49\textwidth]{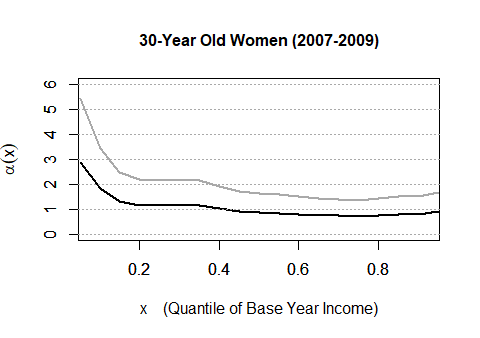}
		\includegraphics[width=0.49\textwidth]{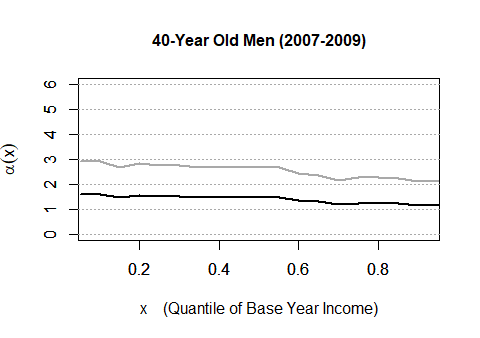}
		\includegraphics[width=0.49\textwidth]{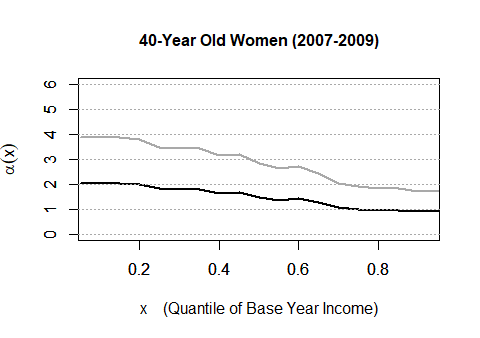}
		\includegraphics[width=0.49\textwidth]{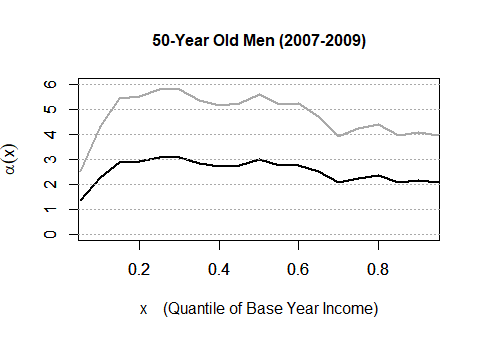}
		\includegraphics[width=0.49\textwidth]{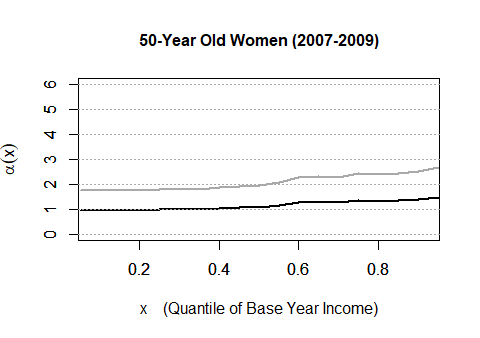}
	\caption{Estimates (black lines) and the one-sided 95\% confidence intervals (gray lines) of the Pareto exponents $\alpha(x_0)$ of the conditional tail risk for men (left) and women (right) based on the PSID in the period 2007--2009. The left (respectively, right) column shows the results for men (respectively, women). The top, middle, and bottom panels show results for 30-, 40- and 50-year-old individuals.}
	\label{fig:psid_2007}
\end{figure}

%
%
%

For 30-year-old men (the top left panel in Figure \ref{fig:psid_2007}), the conditional Pareto exponents (in point estimates) range from 1.7 to 2.9, and the upper bounds of the one-sided 95\% confidence intervals range from 3.1 to 5.3.
Given that income received in 2007 was at or below the median, the conditional Pareto exponent is significantly less than four, implying that the conditional kurtosis of income growth does not exist for these lower-income groups of young men.
The same conclusion also applies to the very top quantiles (top 5 percent).
Overall, the kurtosis barely exists for most of the base-year income levels even if we fail to reject the hypothesis of finite kurtosis.
For 30-year-old women (the top right panel in Figure \ref{fig:psid_2007}), the conditional Pareto exponents (in point estimates) range from 0.7 to 2.9, and the upper bounds of the one-sided 95\% confidence intervals range from 1.4 to 5.4.
For this subpopulation, the conditional Pareto exponent is significantly less than four except for the very bottom quantiles (bottom 5 percent), implying that the conditional kurtosis does not exist. 
Furthermore, given that income received in 2007 was at or above the 15th (respectively, 40th) percentile, the Pareto exponent is significantly less than three (respectively, two), implying that even the conditional skewness (respectively, standard deviation) does not exist.
Comparing the results between men and women at age 30, we observe that women were more vulnerable to income risk than men, except at the very bottom quantiles of the base year income level.

For 40-year-old men (the middle left panel in Figure \ref{fig:psid_2007}), the conditional Pareto exponents (in point estimates) range from 1.2 to 1.6, and the upper bounds of the one-sided 95\% confidence intervals range from 2.1 to 2.9.
Remarkably, the income risk of 40-year-old men is higher than those of 30-year-old men.
For this age group of men, we reject the hypothesis of finite kurtosis at any level of base-year income.
For 40-year-old women (the middle right panel in Figure \ref{fig:psid_2007}), the conditional Pareto exponents (in point estimates) range from 0.9 to 2.0, and the upper bounds of the one-sided 95\% confidence intervals range from 1.7 to 3.9.
For this age group of women, we reject the hypothesis of finite kurtosis at any income level.
Furthermore, given that income received in 2007 was at or above the 50th (respectively, 75th) percentile, the Pareto exponent is significantly less than three (respectively, two), implying that even the conditional skewness (respectively, standard deviation) does not exist.
Comparing the results between men and women at age 30, we observe that men are almost as vulnerable to income risk as women for this middle age group.

For 50-year-old men (the bottom left panel in Figure \ref{fig:psid_2007}), the conditional Pareto exponents (in point estimates) range from 1.3 to 3.1, and the upper bounds of the one-sided 95\% confidence intervals range from 2.5 to 5.8.
For this age group of men, the conditional Pareto exponents are relatively high, and we fail to reject the hypothesis of finite conditional kurtosis except at the very bottom quantiles (bottom 5 percent) of base-year income.
For 50-year-old women (the bottom right panel in Figure \ref{fig:psid_2007}), the conditional Pareto exponents (in point estimates) range from 1.0 to 1.5, and the upper bounds of the one-sided 95\% confidence intervals range from 1.7 to 2.7.
For this age group of women, we reject the hypothesis of finite conditional skewness at any income level.
Furthermore, given that income received in 2007 was at or below the median, the Pareto exponent is significantly less than three, implying that even the conditional skewness does not exist.

\begin{figure}
	\centering
		\includegraphics[width=0.49\textwidth]{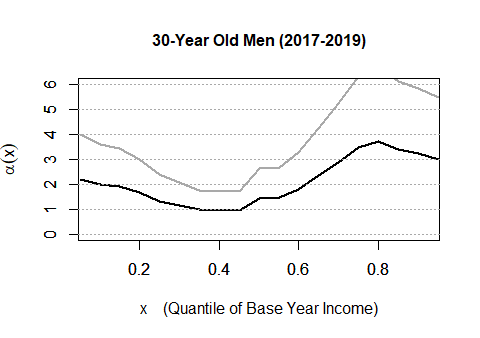}
		\includegraphics[width=0.49\textwidth]{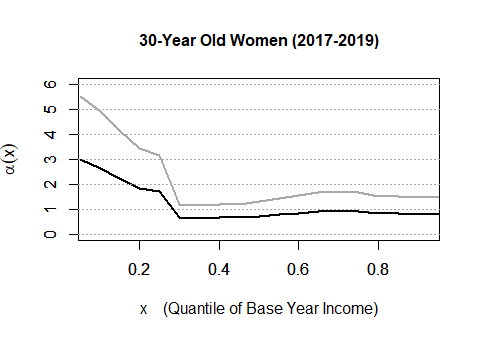}
		\includegraphics[width=0.49\textwidth]{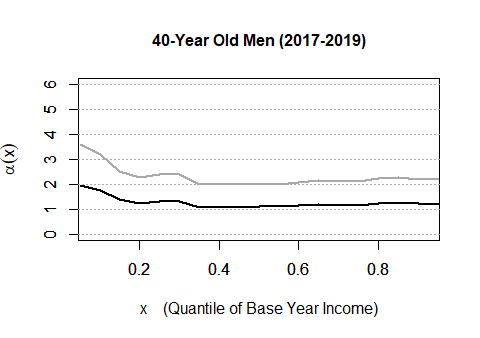}
		\includegraphics[width=0.49\textwidth]{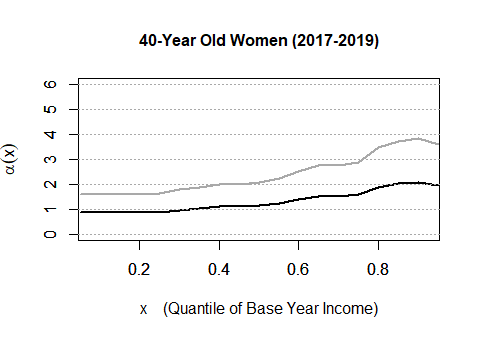}
		\includegraphics[width=0.49\textwidth]{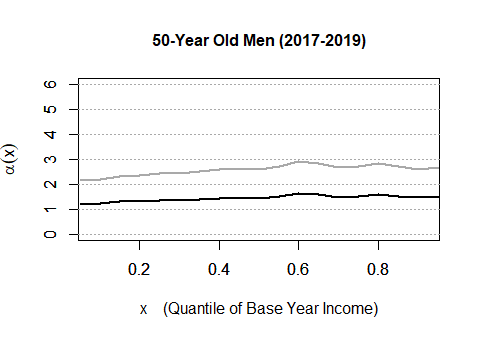}
		\includegraphics[width=0.49\textwidth]{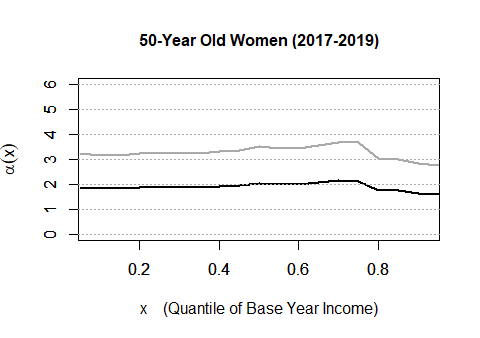}
	\caption{Estimates (black lines) and the one-sided 95\% confidence intervals (gray lines) of the Pareto exponents $\alpha(x_0)$ of the conditional tail risk for men (left) and women (right) based on the PSID in the period 2017--2019. The left (respectively, right) column shows the results for men (respectively, women). The top, middle, and bottom panels show results for 30-, 40- and 50-year-old individuals.}
	\label{fig:psid_2017}
\end{figure}

%
%
%

While our main focus has been on the period of great recession between 2007 and 2009, we next look at the period of positive growth, between 2017 and 2019, ten years later than the baseline period.
Figure \ref{fig:psid_2017} illustrates the results for income growth between 2017 and 2019 for 30-, 40- and 50-year-old individuals, respectively.
These results share similar qualitative patterns to those reported in Figures \ref{fig:psid_2007} as follows. 
First, 30-year-old men at high quantiles of base-year income enjoy less income risk.
Second, 30-year-old women suffer from high income risk except at the bottom quantiles of base-year income. 
Third, 40-year-old men have overall higher income risk than 30-year-old men.
Lastly, and most remarkably, the income risk is not necessarily lower in the period 2017--2019 than in the period 2007--2009, even though the former period enjoyed a positive GDP growth (5.2\% in two years) and the latter period suffered from negative GDP growth ($-$2.7\% in two years).
With all these similarities, there are differences as well -- especially in the graphs that do not resemble across the period between 2017 and 2019 and the period between 2007 and 2009 for 40-year-old women or 50-year-old men.
Overall, there are more similarities than differences despite the contrast between a recession and a positive growth in the US economy.


We also repeat the econometric analysis for other periods, 2009--2011, 2011--2013, 2013--2015 and 2015--2017 to demonstrate the robustness of the observed patterns reported above. 
These figures are similar to Figure \ref{fig:psid_2007} and hence omitted for conciseness. They are available upon request.

To sum up, for the periods 2007--2009 and 2017--2019 we found that:
1) the kurtosis, skewness, and even standard deviation may not exist for the conditional distribution of income growth given certain attributes (age, gender, and income);
2) younger women are more vulnerable to income risk than younger men;
3) middle-aged men are almost as vulnerable as middle-aged women; and
4) these patterns appear both in the period 2007--2009 of great recession and the period 2017--2019 of a positive growth, while there are differences as well.
The first and second points robustly hold in the remaining periods, 2009--2011, 2011--2013, 2013--2015 and 2015--2017.
The third point also continues to hold for the periods, 2009--2011 (to a less extent), 2011--2013 and 2013--2015.
Thus, the fourth point largely extends to the remaining periods, 2009--2011, 2011--2013, 2013--2015 and 2015--2017. 

\newpage
\bibliographystyle{ecta}
\bibliography{bib}
\vspace{0.5cm}
\textbf{Data Citation:}
Office for National Statistics. (2017). New Earnings Survey Panel Dataset, 1975-2016: Secure Access. [data collection]. 7th Edition. UK Data Service. SN: 6706, http://doi.org/10.5255/UKDA-SN-6706-7 
\end{document}